\documentclass{llncs}
\usepackage{llncsdoc}

 \usepackage{graphicx}
\usepackage{color}
\usepackage{setspace}
\usepackage{amsmath,amssymb}
\usepackage{soul}

\usepackage{epstopdf}

\def\set#1{{\{#1\}}}
\def\abs#1{{|\,#1\,|}}

{\begin{list}{}%
         {\setlength{\leftmargin}{#1}}%
         \item[]%
}
{\end{list}}

\newcommand\cG{\mathcal{G}}
\newcommand\cL{\mathcal{L}}

\newcommand\es\emptyset
\newcommand\sse\subseteq
\newcommand\sm\setminus
\newcommand\cP{\mathcal{P}}

\newcommand\N{\mathbb{N}}
\newcommand\Z{\mathbb{Z}}

\newcommand\Sig\Sigma
\newcommand\sig\sigma
\renewcommand\epsilon\varepsilon
\renewcommand\phi\varphi
\renewcommand\theta\vartheta

\newcommand{\trans}{\overset{*}{\underset A\to}}

\newcommand{\sett}[2]{\left\{#1\mathrel{\left|\vphantom{#1}\vphantom{#2}\right.}#2\right\}}

\title{Hypergraph Automata: A Theoretical Model for Patterned Self-assembly}
\date{}
\begin{document}
\author{Lila Kari \and Steffen Kopecki \and Amirhossein Simjour}
\institute{Department of Computer Science, \\
	University of Western Ontario\\
 London, ON, N6A 5B7, Canada}

  \maketitle 
 \begin{abstract}
Patterned self-assembly is a process whereby  coloured tiles  self-assemble to build a rectangular coloured pattern. 
We propose  {\it self-assembly (SA) hypergraph automata} as an automata-theoretic model for patterned
 self-assembly.  We investigate the  computational power of SA-hypergraph automata and show  that for every
recognizable picture language, there exists an SA-hypergraph automaton that accepts this language.
Conversely, we prove that for any restricted SA-hypergraph automaton, there exists a Wang Tile System, 
a model for recognizable picture languages, that accepts the same language. The advantage of 
SA-hypergraph automata over Wang automata, acceptors for the class of recognizable picture languages, is that
 they do not rely on an {\it a priori} defined  scanning strategy.
\end{abstract}  
\section{Introduction}
DNA-based self-assembly is an autonomous  process whereby a disordered system of DNA sequences forms
 an organized structure or pattern as a consequence of Watson-Crick complementarity of DNA sequences,
 without external direction.
A DNA-tile-based self-assembly system starts from  DNA ``tiles'', each of which is formed beforehand from carefully designed single-stranded DNA sequences which bind via Watson-Crick complementarity and ensure the tiles' shape (square) and structure. In particular, 
the sides and interior of the square are double-stranded DNA sequence, while the corners have
 protruding DNA single strands that act as ``sticky ends''.  Subsequently, the individual tiles are mixed together and interact locally via their sticky-ends to form DNA-based {\it supertiles} whose structure is dictated by the base-composition of the individual tiles' sticky ends.
 Winfree \cite{ref38}  introduced the abstract Tile Assembly Model (aTAM)
 as a mathematical model for tile-based self-assembly systems.
 Ma \cite{ref22} introduced  the  patterned self-assembly of single patterns, whereby
  coloured tiles self-assemble to build a particular rectangular {\it coloured pattern}.
 Patterned self-assembly  models a particular type of application in which tiles
 may differ from each other by some distinguishable properties, modelled as colours \cite{motivation01,motivation02}.
Orponen \cite{p001,p002} designed several algorithms to find the minimum tile set required
 to construct one given coloured pattern. Czeizler \cite{p010} proved that this minimization
 problem is NP-hard. 

In this paper, we propose {\it  self-assembly (SA) hypergraph automata} as a
general   model for patterned self-assembly
 and investigate its connections to other models for 
two-dimensional information and computation, such as 2D (picture) languages and Wang Tile Systems.
A  2D (picture) language consists of  2D words (pictures),  defined as  
 mappings $p:[m]\times[n]\rightarrow[k]$ from the points 
 in the two-dimensional space to  a finite alphabet of cardinality $k$.
Here, $[k]$ denotes the set $[k]=\{1,2,\ldots,k\}$. Note that,  if we take the alphabet $[k]$ to
 be a set of colours,  the definition  of a picture is analogous to that of a coloured pattern \cite{ref22}.

Early generating/accepting systems  for 2D languages comprise $2 \times 2$ tiles  \cite{g31},
 2D automata \cite{ref109}, two-dimensional on-line tessellation acceptors \cite{tesselation}, and 2D grammars.
 More recently a  generating system was introduced by Varricchio \cite{wangtileref}  that used {\it Wang tiles}.  
A {\it Wang tile system} \cite{wangtileref} is  a specialized tile-based  model that generates
 the class  of {\it recognizable picture languages}, a subclass of the family of 2D languages.
The class of recognizable picture languages is also accepted by  {\it Wang automata}, a model  introduced in
 \cite{ref105}. Like other automata for 2D languages \cite{g06}, Wang tile automata
 use an explicit pre-defined scanning strategy \cite{ref111}  when reading the input picture and the accepted language 
depends on the scanning strategy that is used. Due to this, Wang automata are a suboptimal model for self-assembly.  
Indeed, if we consider the final supertile as  given, the  order in which tiles  are read   is
 irrelevant. On the other hand, if we consider the self-assembly process which results in the final supertile, an ``order of assembly'' cannot be pre-imposed.
 In contrast to Wang  automata,  SA-hypergraph automata are  scanning-strategy-independent.

SA-hypergraph automata are a  modification of the hypergraph automata introduced
 by Rozenberg \cite{h001} in 1982. An SA-hypergraph automaton (Section~\ref{sec:hypergraph}) accepts a 
language of labelled ``rectangular grid graphs'', wherein the labels are meant to
 capture the notion of colours used in  patterned self-assembly. An SA-hypergraph automaton consists
 of an underlying labelled graph (labelled nodes and edges)
 and a set of {\it hyperedges}, each of which is a subset of the set of nodes of the underlying graph. Intuitively, 
the hyperedges are meant to model tiles or supertiles while the underlying graph describes how these can attach
 to each other, similar to  a self-assembly process. 

 We investigate the computational power of SA-hypergraph automata and prove that for every
 recognizable picture language $L$ there is an SA-hypergraph automaton that accepts  $L$ (Thm.\ \ref{thm:one}). Moreover, 
we prove that for any  restricted SA-hypergraph automaton, there exists a Wang tile system that
 accepts the same language of coloured patterns (Thm.\ \ref{thm:two}). 
Here,  restricted SA-hypergraph automaton means an SA-hypergraph automaton in which certain situations that cannot occur during  self-assembly are explicitly excluded. 
 
\section{Preliminaries}

  A picture (two-dimensional word) $p$ over the alphabet $\Sigma$ is a two dimensional matrix of letters from $\Sigma$. Each element of this matrix is called a $\text{pixel}$. $p_{(i,j)}$ denotes the pixel in the $i$th row and $j$th column of this matrix. Two pixels $p_{(i,j)}$ and $p_{(i',j')}$ are adjacent if $|i-i'| + |j-j'| =1$. The function $w(p)$ denotes the width and $h(p)$ denotes the height of the picture $p$. $\Sigma^{\ast \ast}$ is the set of all pictures over the alphabet $\Sigma$. Let $\#$ be a letter which does not belong to the alphabet $\Sigma$. The {\it framed picture} $\hat{p}$ of $p\in\Sigma^{\ast \ast}$ is defined as:
  
\begin{figure}
  \[
  \hat{p} \ \ = \ \  
  \begin{matrix}
  \# & \# & \cdots & \# & \#\\
  \# & p_{(1,1)} & \cdots & p_{(1,w(p))} & \#\\
  \vdots & \vdots & \ddots & \vdots & \vdots \\
  \# & p_{(h(p),1)} & \cdots & p_{(h(p),w(p))} & \#\\
  \# & \# & \cdots & \# & \#
  \end{matrix}
  \]
\end{figure}  

  A \textit{picture language} (2D language) is a set of pictures over an alphabet $\Sigma$. For example, $L = \{ p \in \Sigma^{\ast \ast} |  \text{ for all } 1\leq i \leq h(p),  p_{(i,1)}=p_{(i,w(p))} \}$
 is the language of all rectangles that have the same first and last column.  

  A function $\delta\colon \N^2\to \N^2$ is a {\it translation function} if there exists $i',j'\in\Z$ such that $\delta(i,j) = (i+i',j+j')$ for all $i,j\in\N$.
  A \textit{subpicture} over $\Sigma$ is a two-dimensional matrix of letters from $\Sigma \cup \{ \text{\it{empty}} \}$. A subpicture $q$ is \textit{connected} if for every pair of pixels $q_{(i',j')},q_{(i,j)}\in \Sigma$ there exists a sequence of pixels $s=\langle s_0, s_1, \ldots , s_n\rangle$ from $q$ such that $s_0 = q_{(i,j)}$ and $s_n = q_{(i',j')}$, for all $0 \leq k < n$, we have $s_k \in \Sigma$. Moreover, $s_k$ and $s_{k+1}$ must be adjacent.
If $p$ is a picture, then $q$ is a subpicture of $p$ if there exists a translation function $\delta$ such that for all $(i,j) \in [h(q)]\times [w(q)]$ we have either $q_{(i,j)}=\text{\it empty}$ or $q_{(i,j)} = p_{\delta(i,j)}$.

  A \textit{picture tile} is a $2 \times 2$ picture (for example \begin{tabular}{|l|l|} \hline $a$ & $b$ \\ \hline $c$ & $d$ \\ \hline \end{tabular}  ). The language defined by a set of picture tiles $\Delta$ over the alphabet $\Sigma \cup \{ \# \}$ is denoted by $\mathcal{L}(\Delta)$ and is defined as the set of all pictures $p \in \Sigma^{\ast \ast}$ such that any $2 \times 2$ subpicture of $\hat{p}$ is in $\Delta$.
Giammarresi \cite{g31} defined a \textit{Picture Tiling System} (PTS)  as a 4-tuple $T = (\Sigma, \Gamma , \Delta, \pi)$, where $\Sigma$ and $\Gamma$ are two finite alphabets, $\Delta$ is a finite set of picture tiles over $\Gamma \cup \{ \# \}$ and $\pi : \Gamma \rightarrow \Sigma$ is a projection. The PTS $T$ recognizes the language $\cL(T) =  \pi(\mathcal{L}(\Delta))$.	
A picture language $L$ is called {\em PTS-recognizable} if there exists a picture tiling system $T$ such that $L=\mathcal{L}(T)$. Figure~\ref{fig:picturetilingsystem} shows an example. 
  
\begin{figure}[ht]
\begin{center}
\includegraphics[width=0.95\textwidth]{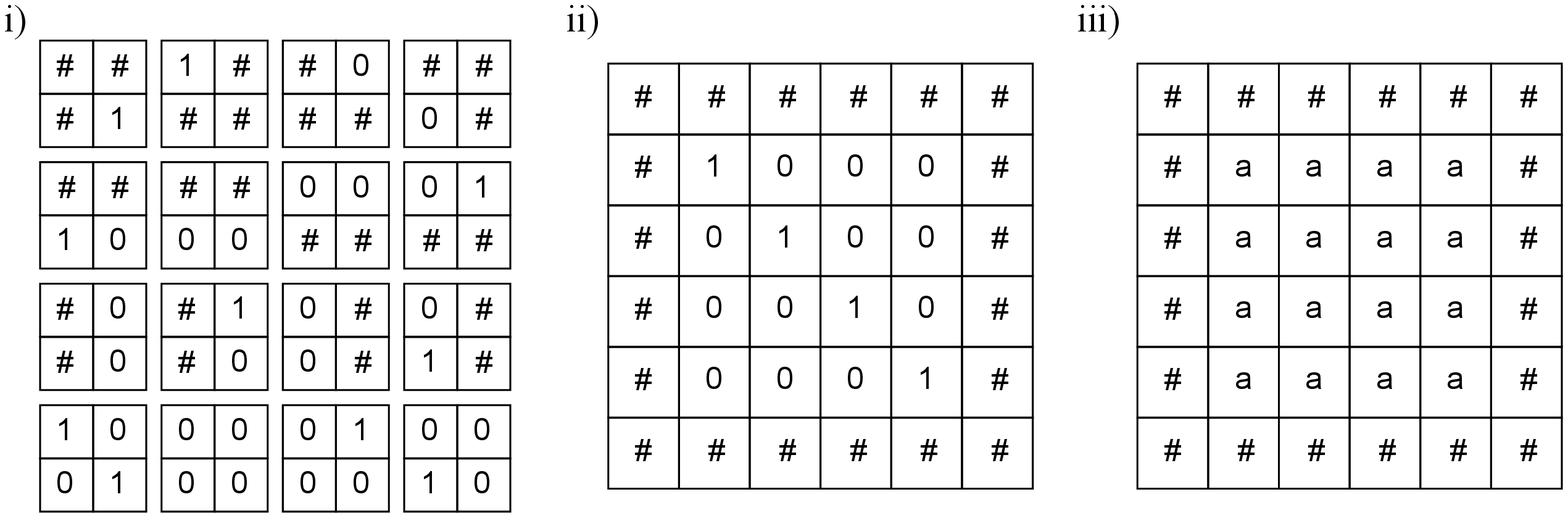}
\caption{Let $T=(\Sigma, T, \Delta, \pi)$ be the picture tile system where $\Sigma=\{a\}$, $\Gamma=\{0,1\}$, $\pi(0)=\pi(1)=a$ and $\Delta$ consists of the 16 different picture tiles in i). The PTS $T$ recognizes the language containing all square pictures $p$ where $w(p) = h(p) \ge 2$ and where every pixel is labelled with $a$. Part ii)\ is an example of framed picture $\hat{p}$ in $\mathcal{L}(\Delta)$, and iii) shows the projection $\pi(\hat{p})$ of the framed picture in part ii).}
\label{fig:picturetilingsystem}
\end{center}
\end{figure}

A equivalent definition of recognizability was proposed using labelled Wang tiles \cite{ref111}. A labelled Wang tile, shortly LWT, is a labelled unit square whose edges may be coloured. Formally, a LWT is a 5-tuple $(c_N, c_E, c_S, c_W, l)$, 
where $l$ belongs to a finite set of labels $\Sigma$ and
$c_N$, $c_E$, $c_S$, and $c_W$ belong to $C\cup\set\#$ where $C$ is a finite set of colours and $\#$ represents an uncoloured edge. Intuitively, $c_N, c_E, c_S$, and $c_W$ represent the colour of the north, east, south, and west edge of the tile, respectively. Labelled Wang tiles cannot rotate. The colours on the north, south, east, and west edges of an LWT $t$ are denoted by $\sigma_N(t)$, $\sigma_S(t)$, $\sigma_E(t)$, and $\sigma_W(t)$, respectively; moreover, $\lambda(t)$ denotes the label of $t$.	

A Wang Tile System (WTS)\cite{wangtileref} is a triple $W=(\Sigma,C,\Theta)$ where $\Sigma$ and $C$ are two finite alphabets (the alphabet of tile labels and the alphabet of colours, respectively) with $\#\notin C$, and $\Theta$ is a finite a set of labelled Wang tiles with labels from $\Sigma$ and colours from $C$.
The WTS $W$ recognizes the picture language $\mathcal{L}(W)$ where the picture $p\in\Sigma^{**}$ belongs to $\mathcal{L}(W)$ if and only if 
there exists a mapping $m\colon [h(p)]\times [w(p)]$ from the pixels of $p$ to tiles from $\Theta$ such that the label of the tile $m(p_{(i,j)})$ is equal to $p_{(i,j)}$;
moreover, this mapping must be {\it mismatch free}.
The mapping $m$ is mismatch free if for two adjacent pixels $p_{(i,j)}$ and $p_{(i+1,j)}$ in $p$ the south edge of $m(p_{(i,j)})$ and the north edge of $m(p_{(i+1,j)})$ are coloured by the same colour from $C$;
for two adjacent pixels $p_{(i,j)}$ and $p_{(i,j+1)}$  in $p$ the east edge of $m(p_{(i,j)})$ and the west edge of $m(p_{(i,j+1)})$ are coloured by the same colour from $C$;
and for every border pixel $p_{(i,j)}$ with $i = 1$, $j=1$, $i=h(p)$, or $j=w(p)$ we require that the north, west, south, or east edge, respectively, of $m(p_{(i,j)})$ is uncoloured.
For a pixel in a corner, e.\,g.\ $p_{(1,1)}$, this implies that two edges are uncoloured.
Let $\bar{p}$ be a two dimensional array of labelled Wang tiles from $\Theta$. We call $\bar{p}$ a Wang tiled version of the picture $p$ if the width and the height of $p$ and $\bar{p}$ are equal, and there exists a mismatch free mapping $m$ such that for any $i$ and $j$ we have $\bar{p}_{(i,j)} = m(p_{(i,j)})$. Two tiles $\bar{p}_{(i,j)}$ and $\bar{p}_{(i',j')}$ are adjacent if the pixels $p_{(i,j)}$ and $p_{(i',j')}$ are adjacent.
A language $L$ is {\em WTS-recognizable} if there exists a Wang tile system $W$ such that $W$ recognizes $L$. 
Figure~\ref{fig:wangtiling} shows an example. 

\begin{figure}[ht]
\begin{center}
\includegraphics[width=0.8\textwidth]{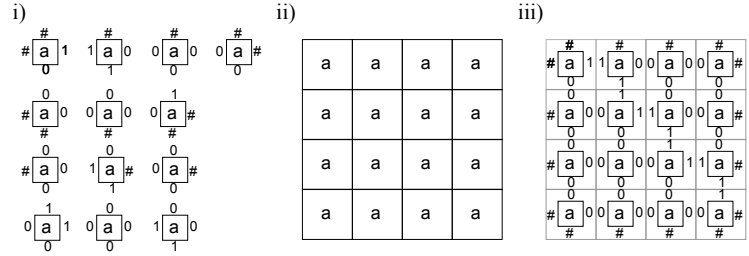}
\caption{Let $W=(\Sigma, C , \Theta)$ be the Wang Tile System where $\Sigma =\{a\}, C=\{0,1\}$ and $\Theta$ consists of the 13 LWTs shown in i). This Wang tile system recognizes the picture language containing all square pictures $p$ with $h(p) = w(p) \ge 3$ and where  every pixel is labelled by $a$. Part ii) is an example picture and iii) shows the Wang tiled version of the picture in part ii).}
\label{fig:wangtiling}
\end{center}
\end{figure}

\begin{proposition}[\cite{g31}]
	A picture language $L$ is PTS-recognizable if and only if it is WTS-recognizable. 
\end{proposition}



A coloured pattern, as defined in \cite{ref22} is the end result of a  self-assembly process 
that starts with a  fixed-size $L$-shaped seed supertile and proceeds as in 
Figure~\ref{fig:simulationmethod}, {\it (i)}, until one coloured rectangle is formed.
 Note that Wang Tile Systems can be seen as generating  
(potentially infinite) languages of such coloured patterns where the $L$-shaped seed is of
 an arbitrary size and is  generated  starting from a single-tiled seed
with uncoloured North and West edges, and is extended by tiles with  uncoloured North
 or West edges, as shown in Figure~\ref{fig:simulationmethod}, {\it (ii)}.

\begin{figure}[ht]
\begin{center}
\includegraphics[scale=.5]{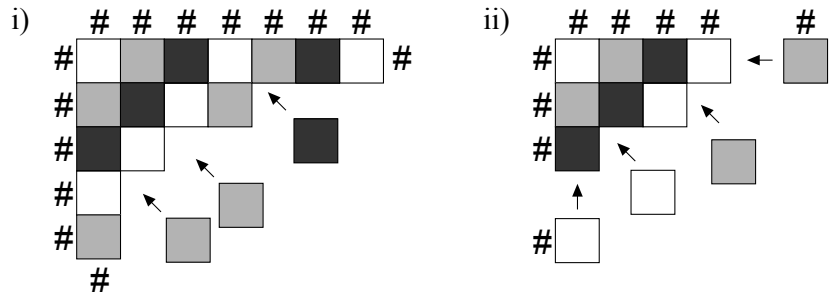}
\caption{{\it (i)} The self-assembly of a single {\it coloured pattern}, starting with a fixed-size $L$-shaped seed.
{\it (ii)}  The process of generating a {\it picture} in the 
language of a {\it Wang Tile System}.} 
\label{fig:simulationmethod}
\end{center}
\end{figure}

\section{Hypergraph Automata}

Let $f\colon A \to B$ be a function and let $A'\sse A$.
The {\em restriction} of $f$ to $A'$ is $f|_{A'}\colon A' \to B$
such that $f|_{A'}(x) = f(x)$ for all $x\in A'$.
For any set $A$ we let $id_A\colon A\to A$ denote the {\em identity}. 
When the set $A$ a is clear from the context, we will omit the subscript and simply write $id$.

Let $\Sigma$ be an alphabet. A {\it pseudo-picture graph} is a directed labelled graph $G = (N,E_v\cup E_h,\pi)$ 
where $N$ is a finite set of nodes,
$E_v, E_h\sse N\times N$ are two sets of edges such that $E_v\cap E_h =\es$,
and $\pi\colon N\to \Sigma$ is the label function.
Edges from $E_v$ and $E_h$ will frequently be denoted by $\overset{v}\longrightarrow$  and $\overset{h}\longrightarrow$, respectively.
The {\em node-induced subgraph} of $G$ by a subset $N'\sse N$ is defined
as the graph $(N',E_v'\cup E_h',\pi|_{N'})$ where
$E_v' = \sett{(x,y)\in E_v}{x,y\in N'}$ and $E_h' = \sett{(x,y)\in E_h}{x,y\in N'}$.
A graph $G'$ is called a {\em full subgraph} of $G$
if for some $N'\sse N$ it is the node-induced subgraph of $G$ by $N'$.

A pseudo-picture graph $G = (N,E_v\cup E_h,\pi)$ is an ($n \times m$-){\it picture graph} (for $n,m\in\N$) if there is a bijection $f_G\colon N \to [N] \times [M]$ such that for $x,y \in N$, we have $(x,y)\in E_v$ if and only if $f_G(x)+(1,0) = f_G(y)$, and $(x,y)\in E_h$ if and only if $f_G(x)+(0,1) = f_G(y)$.
We want to stress that we do not use Cartesian coordinates;
our pictures are defined as matrices, hence, incrementing the first coordinate corresponds to a step downwards, and incrementing the second coordinate corresponds to a step rightwards.
In other words, the nodes of a picture graph $G$ can be embedded in $\N^2$ such that every edge in $E_v$ has length $1$ and points downwards, every edge in $E_h$ has length $1$ and points rightwards, and every two nodes with Euclidean distance $1$ are connected by an edge.
Note that if a pseudo-picture graph is an $n\times m$-picture graph, it cannot be an $n'\times m'$-picture graph with $n\neq n'$ or $m\neq m'$, and the function $f_G$ is unique.
If $G$ is a picture graph, we call $e\in E_v$ a {\em vertical edge} and $e\in E_h$ a {\em horizontal edge}.
The set of all picture graphs is denoted by $\cG$.
Every $n\times m$-picture graph $G = (N,E_v\cup E_h,\pi)$ {\em represents} a picture $p(G)\in \Sig^{**}$ with $h(p(G)) = n$ and $w(p(G)) = m$.
More precisely, for all $(i,j)\in [n] \times [m]$ we let
$p(G)_{(i,j)} = \pi(f_G^{-1}(i,j))$.
Hence, $p\colon \cG \to \Sigma^{**}$ can be seen as a function.

A connected pseudo-picture graph $G$ is called a {\em subgrid} if it is a full subgraph of a picture graph $G'$.
We also say $G$ is  a subgrid of $G'$.
 
A {\em hypergraph} \cite{h001} is a triple $H=(N,E,f)$ where $N$ is the finite set of nodes,
$E$ is the finite set of {\em hyperedges},
and $f\colon E\to \cP(N)$ is a function assigning to each hyperedge
a set of nodes;
the same set of nodes may be assigned to two distinct hyperedges.
For every hyperedge $e \in E$, we let 
\[I_H(e) = \sett{x \in N}{\exists e'\in E\sm\{e\}\colon x\in f(e)\cap f(e')}\]
be the set of {\em intersecting nodes} in $f(e)$.
Rozenberg \cite{h001} introduced {\it hypergraph automata} to describe graph languages. Here, we modified Rozenberg's definition in order to study pseudo-picture graphs. Figure~\ref{fig:example1} shows an one dimensional example of an automaton based on a hypergraph. The formal definition is as follows.

\begin{figure}[ht]
\begin{center}
\includegraphics[width=0.7\textwidth]{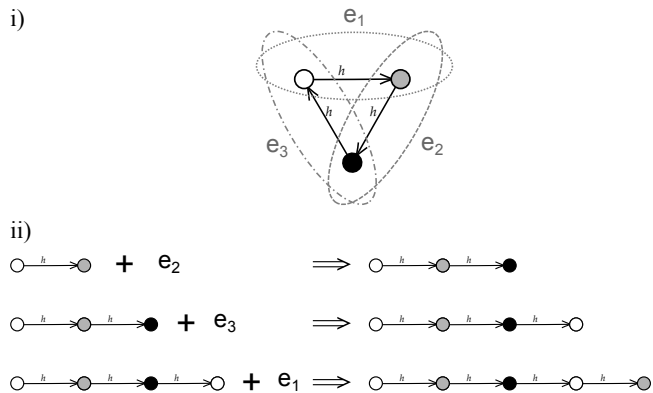}
\caption{The hypergraph in part i) consists of three hyperedges. Intuitively, a derivation of the hypergraph automaton starts from the initial hyperedge $e_1$ and in each step adds the underlying graph of another hyperedge to the current configuration. In part ii) we start with the hyperedge $e_1$, afterwards, the underlying graph of the hyperedge $e_2$ is added to the current configuration. In the next step, the underlying graph of hyperedge $e_3$ is added. In the last step, the hyperedge $e_1$ is reused. Repeating this process, an arbitrary long chain of nodes, forming a white-grey-black pattern, can be constructed using only the three hyperedges.}
\label{fig:example1}
\end{center}
\end{figure}

\begin{definition}
A {\em self-assembly (SA) hypergraph automaton} is a tuple $A = (N,E,f,d,G,E_0)$
where $H=(N,E,f)$ is a hypergraph, called the {\em underlying hypergraph},
$d: E \rightarrow I_H(e) \times I_H(e)$ is the {\em transition function} assigning to each hyperedge $e\in E$ a transition 
$Q_1\to Q_2$ with $Q_1,Q_2\sse I_H(e)$, 
$G$ is a pseudo-picture graph with node set $N$ called the {\em underlying graph}, and
$E_0\sse E$ is the set of {\em initial hyperedges}.
\end{definition}

Every hyperedge $e\in E$ defines a graph $G_e$ which is the subgraph of $G$ induced by $f(e)$.  
For $d(e) = Q_1\to Q_2$ we call $Q_1$ and $Q_2$ the {\em incoming active nodes} and {\em outgoing active nodes} of $G_e$, respectively.
In order for the hypergraph automaton to be well-defined, we require that $G_e$ is connected and that the subgraph of $G_e$ induced by its incoming active nodes is connected, too, for all $e\in E$.
If $e\in E_0$, then $G_e$ is also called an {\em initial graph}.

A {\em configuration} of the hypergraph automaton $A$
is a triple $(M,O,g)$ where
$M=(N_M,E_{M,v}\cup E_{M,h},\pi_M)$ is a subgrid, $O\sse N_M$ is the set of {\em active nodes}, and $g\colon N_M\to N$ is a function such that $\pi_M(x) = \pi(g(x))$ for all $x\in N_M$.
The set $N_M$ consists of (possibly multiple) copies of nodes from $N$ and the function $g$ assigns to each node in $N_M$ its original node in $N$.
An edge $(x,y)\in E_{M,h}$ is a copy of the edge $(g(x),g(y)) \in E_h$ and $(x,y)\in E_{M,v}$ is a copy of the edge $(g(x),g(y)) \in E_v$.
However, for two nodes $x$ and $y$ in $M$, if their originals $g(x)$ and $g(y)$ are connected by a horizontal (or vertical) edge, this does not imply that $x$ and $y$ are connected by a horizontal (or vertical) edge.

Let $(M_1,O_1,g_1)$ be a configuration with $M_1 = (N_1,E_{1,v}\cup E_{1,h},\pi_1)$ and let $e\in E$ be a hyperedge with $d(e) = Q_1 \to Q_2$.
If there exists a non-empty subset $P\sse O_1$ such that $g_1|_P$ forms a graph-isomorphism from the subgraph of $M_1$ induced by $P$ to the subgraph of $G_e$ induced by the incoming active nodes $Q_1$, then the hyperedge $e$ defines a {\em transition} or {\em derivation step}
\[
	(M_1,O_1,g_1) \underset A\to (M_2,O_2,g_2)
\]
where, informally speaking, the resulting graph $M_2$ consists of joining together the graphs $M_1$ and $G_e$ by identifying every node $x \in P$ with the corresponding node $g_1(x) \in Q_1$.
The active nodes $O_2$ in $M_2$ are the active nodes $O_1\sm P$ in $M_1$ plus the outgoing active nodes $Q_2$ in $G_e$, see Figure~\ref{fig:transition:step}.
We also say that $(M_2,O_2,g_2)$ is the result of {\em gluing} the hyperedge $e$ to $(M_1,O_1,g_1)$. 
Formally, the configuration $(M_2,O_2,g_2)$ where $M_2 = (N_2,E_{2,v}\cup E_{2,h},\pi_2)$ is constructed as follows.
Let $N' = \sett{x'}{x\in f(e)\sm Q_1}$ be a set containing a copy of each node from $G_e$ except for the incoming active nodes such that $N'\cap N_1 = \es$.
Let $N_2 = N_1\cup N'$ and let $g_2\colon N_2 \to N$ such that $g_2(x) = g_1(x)$ for $x\in N_1$ and $g_2(x') = x$ for $x'\in N'$.
An edge $(x,y)$ belongs to $E_{2,v}$ if $(x,y)\in E_{1,v}$ or $x,y\in P\cup N'$ and $(g(x),g(y))\in E_v$;
an edge $(x,y)$ belongs to $E_{2,h}$ if $(x,y)\in E_{1,h}$ or $x,y\in P\cup N'$ and $(g(x),g(y))\in E_h$.
Naturally, $\pi_2(x) = \pi(g_2(x))$ for all $x\in N_2$ and $O_2 = ( O_1\sm P ) \cup \sett{x'\in N'}{x\in Q_2}$.
The reflexive and transitive closure of $\underset{A}\to$ is denoted by $\overset{*}{\underset A\to}$ and called a {\em derivation}.

\begin{figure}[ht]
\begin{center}
\includegraphics[scale=.5]{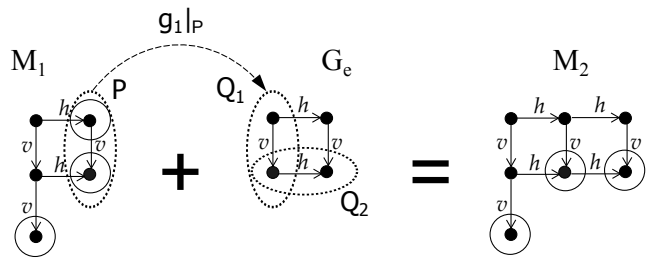}
\caption{A transition $(M_1,O_1,q_1)\rightarrow_A (M_2,O_2,q_2)$ joins together the graphs $M_1$ and $G_e$  by  identifying every node  $x\in P$  with the corresponding node $g_1(x)\in Q_1$. The set $O_2$ of the active nodes of the
new configuration $M_2$ consists of the nodes of the union of 
the active nodes  in $O_1\sm P$  with the outgoing active nodes $Q_2$ of $G_e$. The active nodes of $M_1$ and $M_2$
 are represented as circled nodes.
}

\label{fig:transition:step}
\end{center}
\end{figure}

For $e\in E_0$ we let $O_e$ such that $d(e) = Q_1\to O_e$ and we call
the configuration $(G_e,O_e,id)$ an {\em initial configuration} of $A$.
A {\em final configuration} is a configuration $(M,\es,g)$ without active nodes.
The graph language accepted by the SA-hypergraph automaton $A$ is
\[
	\cL(A) = \sett{M \in \cG}{\exists e\in E_0\colon (G_e,O_e,id)\trans(M,\es,g)}.
\]
Note that $\cL(A)$ contains picture graphs only.
The {\em picture language associated to the graph language} $\cL(A)$ is the language $p(\cL(A))$.

\begin{remark}\label{rem:subgrid}
Since we only talk about picture graphs, we can assume that for every hyperedge $e \in E$ the underlying graph $G_e$ is a subgrid, or $e$ can be removed from the set $E$.
\end{remark}

\begin{example}
Figure~\ref{fig:patternHGA} shows an example of a self-assembled coloured pattern and  an SA-hypergraph automaton that accepts
that pattern. Part i) depicts a coloured self-assembled pattern. Parts ii) and iii) together depict the underlying graph of
the SA-hypergraph automaton that constructs the same pattern.
\begin{figure}[ht]
\begin{center}
\includegraphics[width=1.0\textwidth]{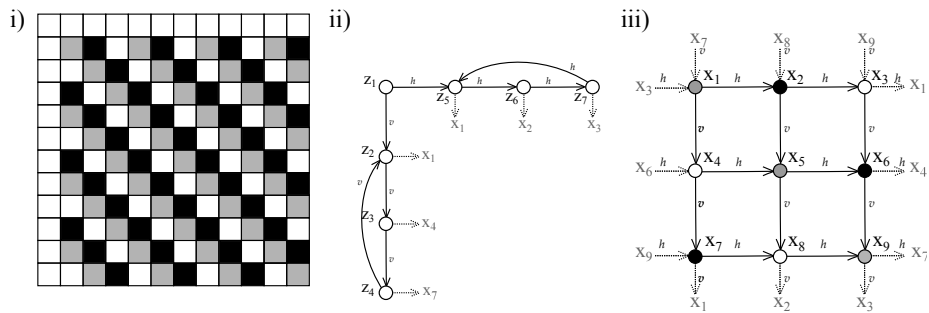}
\caption{
Part i) shows an example of coloured self-assembled pattern. Parts ii) and iii) 
together depict the underlying graph of the SA-hypergraph automaton that constructs the same pattern. Part ii)
 constructs the white top row and white left  column,  and part iii) constructs the coloured pattern.
}
\label{fig:patternHGA}
\end{center}
\end{figure}

The SA-hypergraph automaton for the example in Figure 5 is defined as follows. 
The SA-hypergraph automaton is 
$A = (N,E,f,d,G,E_0)$, 
where 
\begin{itemize}
\item $N=\{x_1,x_2,\dots,x_9,z_1, z_2, \dots, z_7\}$, 
\item $E=\{e_1,e_2, \dots,e_{16}\}$, 
\item function $f$ is defined such that

\begin{align*}
&f(e_1) = \{ x_1, x_2, x_4, x_5 \}, &
&f(e_2) = \{ x_2, x_3, x_5, x_6 \}, &
&f(e_3) = \{ x_3, x_1, x_6, x_4 \}, \\
&f(e_4) = \{ x_4, x_5, x_7, x_8 \}, &
&f(e_5) = \{ x_5, x_6, x_8, x_9 \}, &
&f(e_6) = \{ x_6, x_4, x_9, x_7 \}, \\
&f(e_7) = \{ x_7, x_8, x_1, x_2 \}, &
&f(e_8) = \{ x_8, x_9, x_2, x_3 \}, &
&f(e_9) = \{ x_9, x_7, x_3, x_1 \}, \\
&f(e_{10}) = \{ z_1,z_5,z_2,x_1\},  &
&f(e_{11}) = \{ z_5,z_6, x_1,x_2\}, &
&f(e_{12}) = \{ z_6,z_7,x_2,x_3\}, \\
&f(e_{13}) = \{ z_7,z_5,x_3,x_1\},  &
&f(e_{14}) = \{ z_2,x_1,z_3,x_4\},  &
&f(e_{15}) = \{ z_3,x_4, z_4,x_7\}, \\
&f(e_{16}) = \{ z_4,x_7,z_2,x_1\}.
\end{align*}
\item 
For each hyperedge in ii),  the function $d$ describing the active areas where we can glue new hyperedges
 is defined as to build a horizontal (vertical) chain of nodes  that models the top row (left column) of tiles.
 \begin{align*}
 & d(e_{11}) = \{ z_5, x_1\} \rightarrow \{ z_6,x_1,x_2\}, \; & 
 & d(e_{12}) = \{ z_6,x_2\} \rightarrow \{ z_7,x_2,x_3\}, \\
 & d(e_{13}) = \{ z_7,x_3\} \rightarrow \{ z_5,x_1,x_3\}, \; &
 & d(e_{14}) = \{ z_2, x_1\} \rightarrow \{ z_3,x_1,x_4\}, \\
 & d(e_{15}) = \{ z_3,x_4\} \rightarrow \{ z_4,x_4,x_7\}, \; &
 & d(e_{16}) = \{ x_7,z_4\} \rightarrow \{ z_2,x_1,x_2\}.
  \end{align*}
The backward edges e.g.\ $(x_3,x_1)$, $(x_4,x_6)$,  $(x_7,x_9)$, and $(z_7,z_5)$, make it possible to reuse the hyperedges to
 build a periodic pattern. 

For each hyperedge in iii), the  function $d$  changes the active input nodes  (top-left, bottom-left,
 and top-right) to  the new set of active nodes  (top-right, bottom-left, and bottom-right), signifying the change of 
 the places where the new hyperedges can be glued.

\begin{align*}
&d(e_1) = \{ x_1, x_2, x_4\} \rightarrow \{ x_2,x_4,x_5\},  &
&d(e_2) = \{ x_2, x_3, x_5\} \rightarrow \{ x_3,x_5,x_6\},  \\
&d(e_3) = \{ x_3, x_1, x_4\} \rightarrow \{ x_1,x_6,x_4\},  &
&d(e_4) = \{ x_4, x_5, x_7\} \rightarrow \{ x_5,x_7,x_8\},  \\
&d(e_5) = \{ x_5, x_6, x_8\} \rightarrow \{ x_6,x_8,x_9\},  &
&d(e_6) = \{ x_6, x_4, x_9\} \rightarrow \{ x_4,x_9,x_7\},  \\
&d(e_7) = \{ x_7, x_8, x_1\} \rightarrow \{ x_8,x_1,x_2\},  &
&d(e_8) = \{ x_8, x_9, x_2\} \rightarrow \{ x_9,x_2,x_3\},  \\
&d(e_9) = \{ x_9, x_7, x_3\} \rightarrow \{ x_7,x_3,x_1\},  &
&d(e_{10}) = \{ z_1, z_5, z_2\} \rightarrow \{ z_5,x_1,z_2\}.
\end{align*}

\item
Parts ii) and iii)  depict the underlying graphs of the white $\Gamma$-shaped top and left border of the pattern, 
and the white-grey-black part of the pattern respectively. 
\item $E_0 = \{ e_{10} \}$
\end{itemize}

The SA-hypergraph automaton $A$ starts from the top-left white tile, corresponding to $E_0 = \{ e_{10} \}$. Afterwards,
 the automaton  continues the construction with the hyperedges in the top row or the left column.
 The construction of the white-grey-black part starts after the construction  of the white top row and left column. 
Figure~\ref{fig:example2} shows an example of possible transitions of the SA-hypergraph automaton $A$.

\begin{figure}[ht]
\begin{center}
\includegraphics[width=0.93\textwidth]{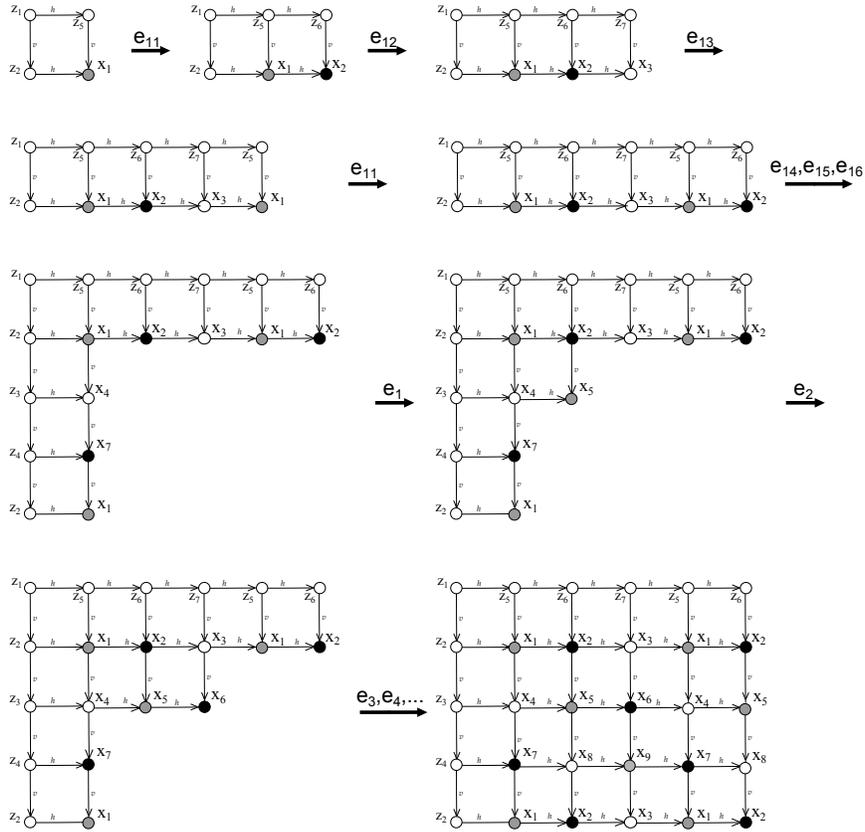}
\caption{
In this example, the construction of a picture graph from Figure 5 is explained. At each step, one hyperedge or a sequence of hyperedges is glued. 
}
\label{fig:example2}
\end{center}
\end{figure}
\end{example}

The concept of hypergraph automata has been introduced by Rozenberg in 1982 \cite{h001}.
Our definition of SA-hypergraph automata is a variant of the original definition with the following modifications.
Firstly, we start from a set of initial graphs whereas the original definition used a single initial graph.
For unlabelled graphs both models are capable of accepting the same class of graph languages, as long as one makes an exception for the empty graph.
However, for labelled graphs a single initial graph is not sufficient;
e.\,g., if a language $L$ of labelled graphs contains one graph $A$ where every node is labelled by $a$ and one graph $B$ where every node is labelled by $b$, then $L$ cannot be generated from the same initial graph since $A$ and $B$ do not have a common non-empty isomorphic subgraph.
Secondly, we use final configurations in order to accept only some of the graphs that
can be generated by rules from the initial graph.
In the original definition, for simplicity, final configurations were omitted and every graph which can be generated from the initial graph belonged to the accepted language.
Thirdly, it seemed more convenient to us to use the notion of active nodes rather than active intersections.

\section{Some Examples}
In this section, we provide three example SA-hypergraph automata and illustrate their relation to self-assembly systems.
Our findings, presented in Section~\ref{sec:results}, do not build upon this section. 
In all examples, every node in the underlying graph has a distinct colour which, for simplicity, is the same as the identifier of the node.

The following examples shows a SA-hypergraph automaton to accept the pictures in Figure~\ref{fig:examples01} part $a)$. This example shows that SA-hypergraph automata can accept a picture language with a simple description. The SA-hypergraph automaton in this example has 8 nodes and 3 hyperedges; the equivalent tile system needs 8 tile types. 

\begin{example}
The SA-hypergraph automaton for the example in Figure ~\ref{fig:examples01} is defined as follows. 
The SA-hypergraph automaton is 
$A = (N,E,f,d,G,E_0)$, 
where 
\begin{itemize}
\item $N=\{x_1,x_2,x_3,x_4,x_5,x_6,x_7,x_8\}$, 
\item $E=\{e_1,e_2, e_{3}\}$, 
\item function $f$ is defined such that
\begin{align*}
&f(e_1) = \{ x_1, x_2,x_3,x_4\}, \\
&f(e_2) = \{ x_3,x_4,x_5,x_6\}, \\
&f(e_3) = \{  x_3, x_4 ,x_7,x_8\}
\end{align*}
\item function $d$ is defined such that
\begin{align*}
&d(e_1) = \{ x_1, x_2\} \rightarrow \{x_3,x_4\}, \\
&d(e_2) = \{ x_3,x_4 \} \rightarrow \{ \}, \\
&d(e_3) = \{  x_3, x_4 \} \rightarrow \{ \}
\end{align*}
\item underlying graph is shown in Figure  ~\ref{fig:examples01} part $b$.
\item $E_0 = \{e_1\}$

\end{itemize}
\label{exmple:examples01}
\end{example}

\begin{figure}[ht]
\begin{center}
\includegraphics[scale=.3]{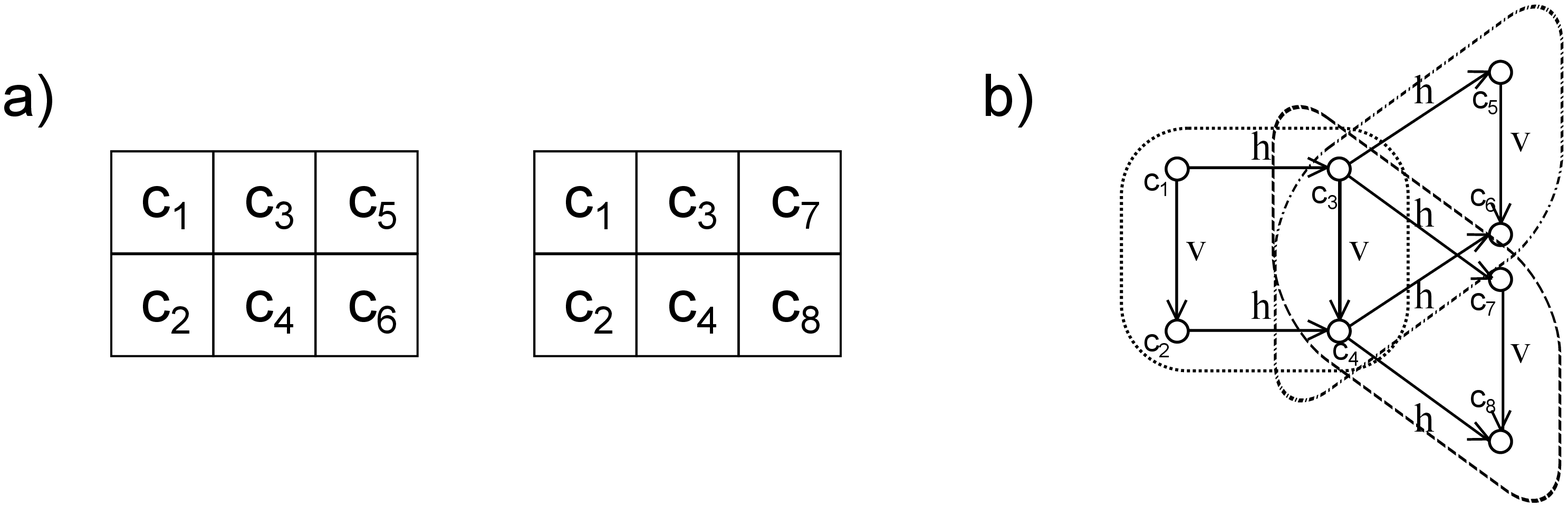}
\caption{Part a) shows an example of language of coloured self-assembled patterns. Parts b) depicts the underlying graph of the SA-hypergraph automaton that constructs the same pattern.}
\label{fig:examples01}
\end{center}
\end{figure}

Example ~\ref{exmple:examples02} shows a simple picture language containing two 2D-words. 
The SA-hypergraph automaton uses two overlapping hyperedges with different active inputs and outputs. Therefore, the number of nodes in this SA-hypergraph automaton will be less than the number tiles in a tile assembly system which recognizes the same language. The SA-hypergraph automaton in this example has 4 nodes and 3 hyperedges. An equivalent tile assembly system needs at least 6 tile types.  

\begin{example}
The SA-hypergraph automaton for the example in Figure~\ref{fig:examples02} is defined as follows. 
The SA-hypergraph automaton is 
$A = (N,E,f,d,G,E_0)$, 
where 
\begin{itemize}
\item $N=\{x_1,x_2,x_3,x_4\}$, 
\item $E=\{e_1,e_2, e_{3}\}$, 
\item function $f$ is defined such that
\begin{align*}
&f(e_1) = \{ x_1, x_2\}, \\
&f(e_2) = \{ x_1,x_2\}, \\
&f(e_3) = \{  x_1, x_2 ,x_3,x_4\}
\end{align*}
\item function $d$ is defined such that
\begin{align*}
&d(e_1) = \{ x_1\} \rightarrow \{ \}, \\
&d(e_2) = \{ x_1\} \rightarrow \{x_1,x_2 \}, \\
&d(e_3) = \{  x_1, x_2 \} \rightarrow \{ \}
\end{align*}
\item underlying graph is shown in Figure~\ref{fig:examples02} part $b)$.
\item $E_0 = \{e_1,e_2\}$

\end{itemize}

\label{exmple:examples02}
\end{example}

\begin{figure}[ht]
\begin{center}
\includegraphics[scale=.3]{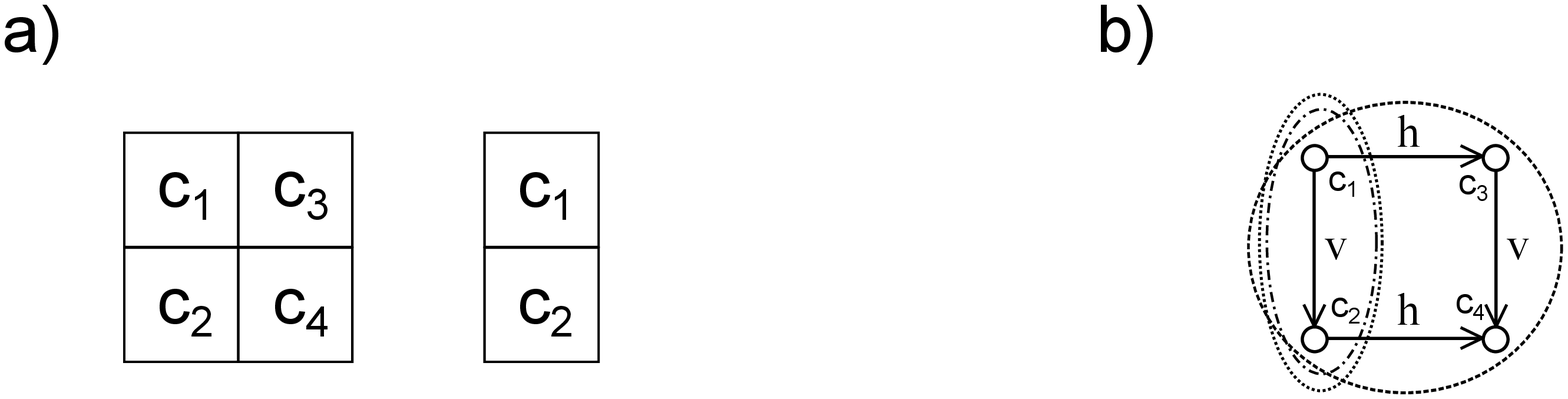}
\caption{Part a) shows an example of language of coloured self-assembled patterns. Parts b) depicts the underlying graph of the SA-hypergraph automaton that constructs the same pattern.}
\label{fig:examples02}
\end{center}
\end{figure}
Example ~\ref{exmple:examples03} shows a language with an infinite number of one dimensional pictures. 
The SA-hypergraph automaton uses three hyperedges to build the chain, moreover, one more heyperedge is used to make the final configurations. Therefore, the number of nodes in this SA-hypergraph automaton will be less than the number tiles in a tile assembly system which recognizes the same language. The SA-hypergraph automaton in this example has 3 nodes and 4 hyperedges. Whereas an equivalent tile assembly system needs at least 5 tile types (one tile type to start, 3 tile type to build the chain, and one tile type to stop).  

\begin{figure}[ht]
\begin{center}
\includegraphics[scale=.3]{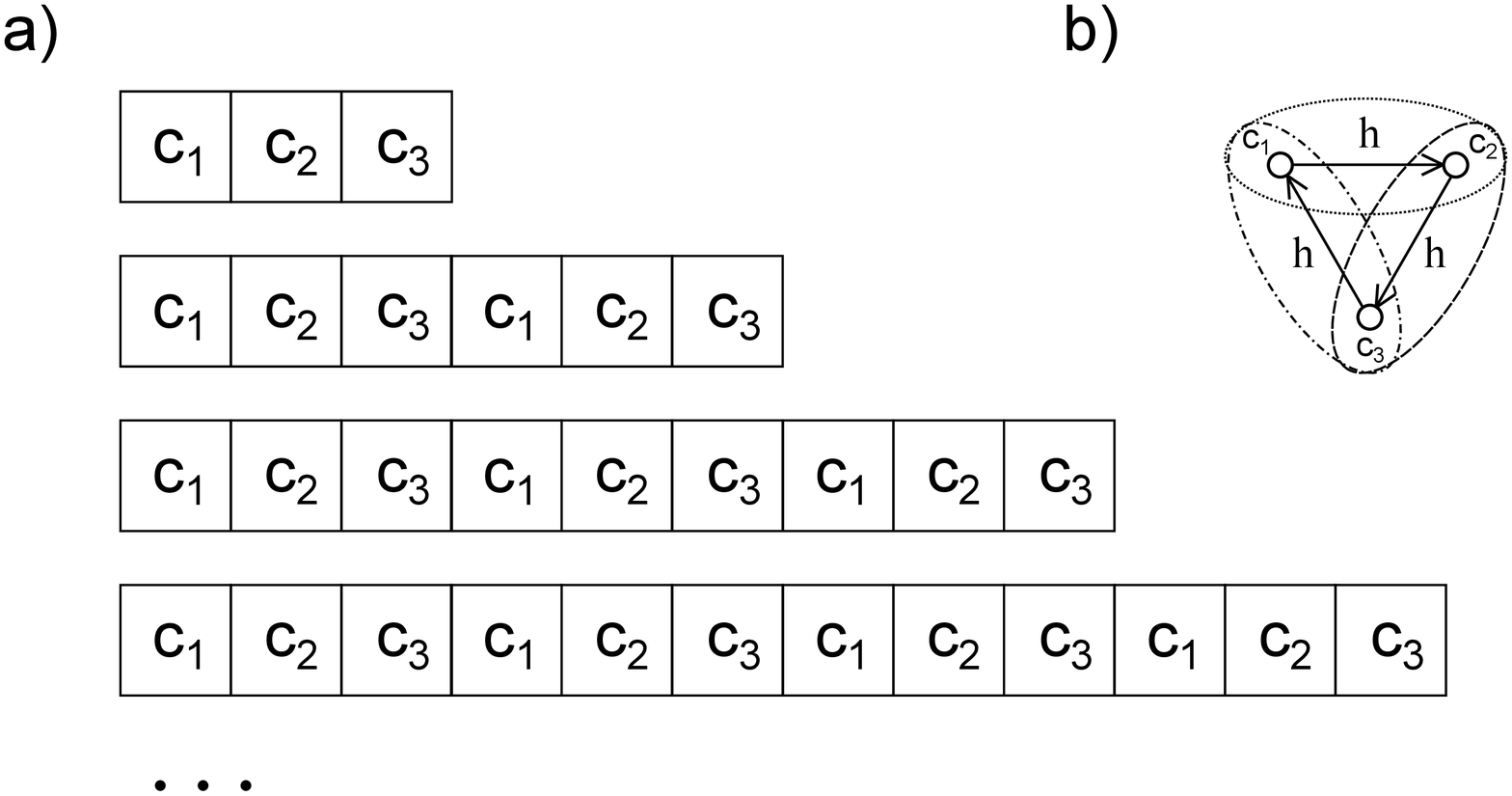}
\caption{Part a) shows an example of language of coloured self-assembled patterns. Parts b) depicts the underlying graph of the SA-hypergraph automaton that constructs the same pattern.}
\label{fig:examples03}
\end{center}
\end{figure}

\begin{example}
The SA-hypergraph automaton for the example in Figure ~\ref{fig:examples03} is defined as follows. 
The SA-hypergraph automaton is 
$A = (N,E,f,d,G,E_0)$, 
where 
\begin{itemize}
\item $N=\{x_1,x_2,x_3\}$, 
\item $E=\{e_1,e_2, e_{3},e_4\}$, 
\item function $f$ is defined such that
\begin{align*}
&f(e_1) = \{ x_1, x_2\}, \\
&f(e_2) = \{ x_2,x_3\}, \\
&f(e_3) = \{  x_3, x_1\}
&f(e_4) = \{ x_2,x_3\}, \\
\end{align*}
\item function $d$ is defined such that
\begin{align*}
&d(e_1) = \{ x_1\} \rightarrow \{ x_2\}, \\
&d(e_2) = \{ x_2\} \rightarrow \{x_3\}, \\
&d(e_3) = \{  x_3\} \rightarrow \{x_1 \}
&d(e_4) = \{ x_2\} \rightarrow \{ \}, \\
\end{align*}
\item underlying graph is shown in Figure  ~\ref{fig:examples03} part $b)$.
\item $E_0 = \{e_1\}$

\end{itemize}

\label{exmple:examples03}
\end{example}
\section{Hypergraph Automata for Picture Languages}\label{sec:results}

In this section, we establish a strong connection between (WTS-)recognizable picture languages and picture graph languages that can be accepted by SA-hypergraph automata.
We prove that the self-assembly of a Wang Tile System can be simulated by an SA-hypergraph automaton, see Theorem~\ref{thm:one}.
The main idea is to start the tiling in the top left corner of a tiled picture and then extend the tiled picture downwards and rightwards, just as in Figure~\ref{fig:simulationmethod}.
Our converse result is slightly weaker: the picture language $L=p(\cL(A))$, associated to the graph language accepted by an SA-hypergraph automaton $A$, is WTS-recognizable if $A$ does not contain a {\em strong loop}, see Theorem~\ref{thm:two}.
The restriction for $A$ not to contain a strong loop is a natural assumption as strong loops cannot be used in any derivation that accepts a picture graphs.

\begin{theorem}\label{thm:one}
For any recognizable picture language $L$ there is a SA-hypergraph automaton $A$
such that the picture language associated to the graph language $\cL(A)$ is $L$.
\end{theorem}

\begin{proof}
Let $V=(\Sigma,C',\Theta')$ be a Wang Tile System that recognizes the picture language $L$, that is $L = \cL(V)$.
We will slightly modify the WTS $V$ such that it fulfils a certain property as described in the following.
We define a WTS $W=(\Sigma,C,\Theta)$ which recognizes $L$ and such that any two copies of a tile $t\in\Theta$ in a tiling of $W$ must have a row- and a column-distance which is a multiple of $3$.
More precisely, for a Wang tiled version $\bar p$ of a picture $p\in\cL(W)$ where a tile $t\in\Theta$ appears at two positions $t= \bar p_{(i,j)} = \bar p_{(i',j')}$, we have that $3$ divides $|i-i'|$ as well as $ |j-j'|$.
This is achieved by using $9$ copies of every tile from $V$ in $W$;
we let $\Theta = \Theta'\times\set{0,1,2}\times\set{0,1,2}$.
We will ensure that a tile $(t,i,j)\in\Theta$ can only appear at position $(i',j')$ if $i = i'\bmod 3$ and $j = j'\bmod 3$.
This property is achieved by defining the glues as $C = C'\times\set{0,1,2}$, and  for  $t = (s,i,j)\in \Theta'\times\set{0,1,2}\times\set{0,1,2}$ we let
\begin{align*}
	\lambda(t) &= \lambda(s), &
	\sigma_S(t) &= (\sigma_S(s),i), &
	\sigma_N(t) &= (\sigma_N(s),(i-1)\bmod 3), \\
	&&
	\sigma_E(t) &= (\sigma_E(s),j), &
	\sigma_W(t) &= (\sigma_W(s),(j-1)\bmod 3).
\end{align*}
Note that a tiled picture $\bar p$ of $W$ can be converted into a tiled picture $\bar q$ of $V$ such that the corresponding pictures $p$ and $q$ coincide by applying the mapping $(t,i,j)\mapsto t$ to every tile in $\bar p$.
Vice versa, a tiled picture $\bar q$ in $V$ can be converted into a tiled picture $\bar p$ in $W$ such that the corresponding pictures $p$ and $q$ coincide by applying the mapping $\bar p_{(i,j)} \mapsto (p_{(i,j)},i\bmod 3,j\bmod 3)$ to every position in $\bar p$.

The modification of $V$ will become of importance later in the proof:
We need to ensure that for a $2\times 2$ square of matching tiles $t_1,t_2,t_3,t_4$, it is not possible to directly attach another copy of any of $t_1,t_2,t_3,t_4$ to this square.

We will define a SA-hypergraph automaton $A= (N,E,f,d,G,E_0)$ which simulates the assembly of a tiled picture from $L = \cL(W)$ as described in Figure~\ref{fig:simulationmethod}. 

Let $N$ be a set of nodes such that $\abs{N} = \abs{\Theta}$ and let $\theta\colon N\to\Theta$ be a bijection.
For each node $x\in N$ there is a {\em corresponding tile} $\theta(x)$ and vice versa.
Let $N_T$, $N_R$, $N_B$, $N_L$ be the set of nodes which correspond to tiles on the top, right, bottom, left border of a tiled picture, respectively:
\begin{align*}
	N_T &= \sett{x\in N}{\sigma_N(\theta(x)) = \#}, &
	N_R &= \sett{x\in N}{\sigma_E(\theta(x)) = \#}, \\
	N_B &= \sett{x\in N}{\sigma_S(\theta(x)) = \#}, &
	N_L &= \sett{x\in N}{\sigma_W(\theta(x)) = \#}.
\end{align*}

Let $G=(N,E_v\cup E_h,\pi)$ be the underlying graph of $A$.
The label function $\pi$ is naturally defined as
$\pi(x) = \lambda(\vartheta(x))$ for $x\in N$.
For all nodes $x,y\in N$ there is an edge $(x,y)\in E_h$ if and only if $\sigma_E(\vartheta(x)) = \sigma_W(\vartheta(y))\neq \#$ and either
$x,y\in N\sm(N_T\cup N_B)$ or $x,y\in N_T$ or $x,y\in N_B$;
there is an edge $(x,y)\in E_v$ if and only if $\sigma_S(\vartheta(x)) = \sigma_N(\vartheta(y))\neq \#$ and either $x,y\in N\sm(N_L\cup N_R)$ or $x,y\in N_L$ or $x,y\in N_R$.
This means if the east edge of a tile $t$ can attach to the west edge of tile $s$, then their corresponding nodes $x=\theta^{-1}(t)$ and $y=\theta^{-1}(s)$ are connected by an $h$-edge $(x,y)\in E_h$.
Analogously, if the south edge of a tile $t$ can attach to the north edge of tile $s$, then their corresponding nodes $x=\theta^{-1}(t)$ and $y=\theta^{-1}(s)$ are connected by an $v$-edge $(x,y)\in E_v$.

If $N_T\cap N_B \neq \es$ or $N_R \cap N_L \neq \es$, the language $\cL(W)$ possibly contains pictures $p$ with $h(p) = 1$ or $w(p)=1$, respectively, which can be seen as one-dimensional pictures.
These pictures have to be treated separately.
For now we assume that $N_T\cap N_B = N_R\cap N_L=\es$.

The hyperedges $E$ and the transition function $d$ define the possible transitions of $A$.
In every transition we add exactly one node to the graph of a configuration of $A$.
Our naming convention is that $x$ is the node which is attached in the derivation step and $y, y_1,y_2,y_3$ are incoming active nodes of the hyperedge.
Every graph containing only one node which corresponds to a tile in the top left corner is an initial graph.
In order to construct a picture graph which represents a picture in $\cL(W)$ we introduce three types of transitions, see Figure~\ref{fig:hyperedges}.
The transitions of type~I generate the top row of the graph and transitions of type~II generate the left column of the graph; both transition types keep every generated node active.
Transitions of type~III generate the rest of the graph: A node is attached if it has a matching east neighbour ($y_1$), a matching north neighbour ($y_3$), and these two nodes are connected by another node ($y_2$); unless we reach the right or bottom border of the graph the nodes $x$, $y_1$, and $y_3$ are active after using the transition.

\begin{figure}[ht]
\begin{center}
\includegraphics[scale=.5]{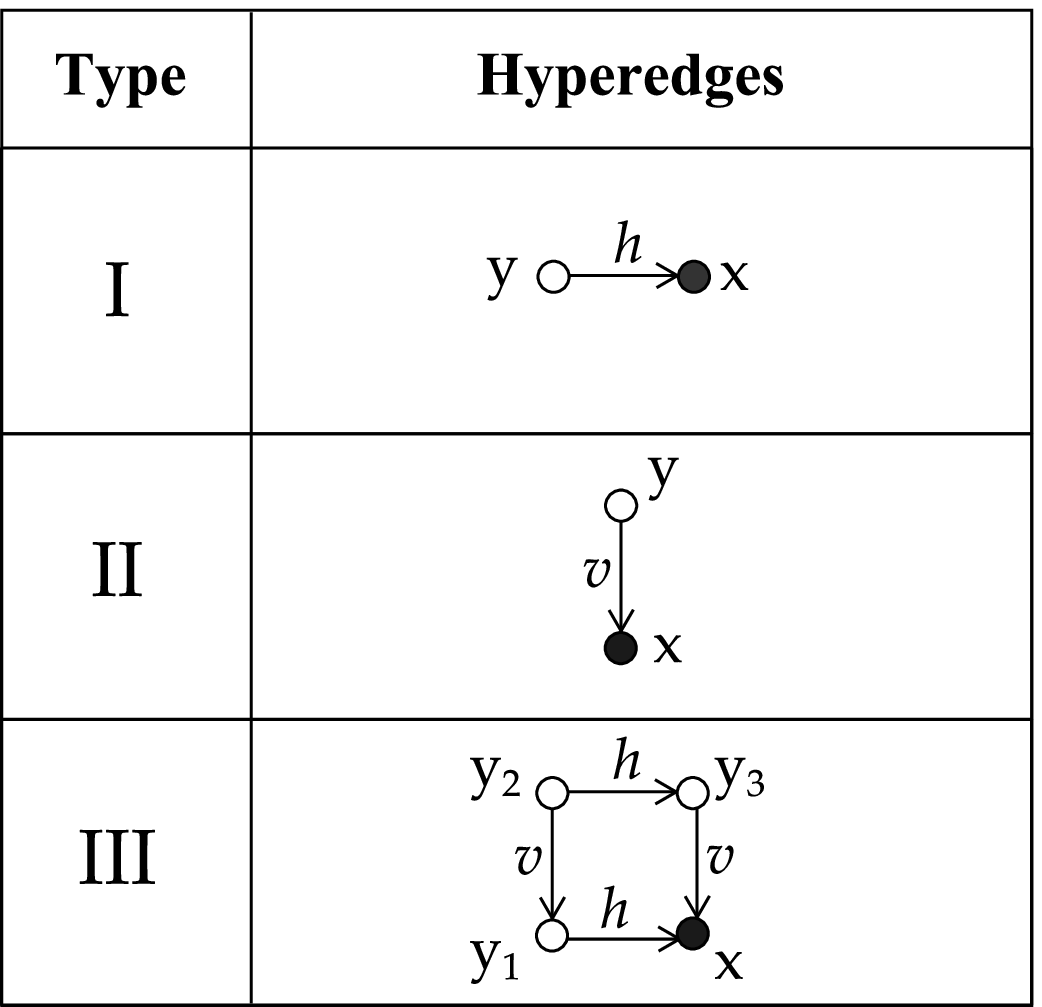}
\caption{The hyperedges in the SA-hypergraph automaton $A$ induce three different types of graphs. White nodes represent incoming active nodes of the hyperedges.}
\label{fig:hyperedges}
\end{center}
\end{figure}

Formally, we define the set of hyperedges $E$, the set of initial edges $E_0$, the function $f$, and the transition function $d$ as following:

{\bf Initial graphs: }
For each $x\in N_T\cap N_L$, corresponding to a tile in the top left corner,
we define a hyperedge $e_x\in E_0\sse E$ with associated nodes $f(e_x) = \set x$
and the transition function $d(e_x) = \es\to\set{x}$.

{\bf Type I: }
For all nodes $x,y\in N_T$, in the top row, such that $(x,y)\in E_h$,
we define a hyperedge $e_{x,y}\in E$ with associated nodes $f(e_{x,y}) = \set{x,y}$
and the derivation function $d(e_{x,y}) = \set y \rightarrow \set{x,y}$.

{\bf Type II: }
For all nodes $x,y\in N_L$, in the left column, such that $(x,y)\in E_v$,
we define a hyperedge $e_{x,y}\in E$ with associated nodes $f(e_{x,y}) = \set{x,y}$
and the derivation function $d(e_{x,y}) = \set y \rightarrow \set{x,y}$.

{\bf Type III: }
For all nodes $x\in N\sm(N_T\cup N_L)$ and $y_1,y_2,y_3\in N$
such that $(y_2,y_1),(y_3,x)\in E_v$ and $(y_2,y_3),(y_1,x)\in E_h$,
we define a hyperedge $e_{x,y_1,y_2,y_3}\in E$
with associated nodes $f(e_{x,y_1,y_2,y_3}) = \set{x,y_1,y_2,y_3}$ and
the derivation function
\begin{enumerate}
	\item $d(e_{x,y_1,y_2,y_3})=\set{y_1,y_2,y_3}\to \es$ if $x\in N_B\cap N_R$,
		(bottom right corner)
	\item $d(e_{x,y_1,y_2,y_3})=\set{y_1,y_2,y_3}\to \set{x,y_3}$ if $x\in N_B\sm N_R$,
		(bottom row)
	\item $d(e_{x,y_1,y_2,y_3})=\set{y_1,y_2,y_3}\to \set{x,y_2}$ if $x\in N_R\sm N_B$,
		(right column)
	\item $d(e_{x,y_1,y_2,y_3})=\set{y_1,y_2,y_3}\to \set{x,y_2,y_3}$ otherwise.
\end{enumerate}

Consider the graph $G_e$ which is induced by the hyperedge $e\in E$.
Depending on the type of the hyperedge $e$, the graph $G_e$ contains at least the edges shown in Figure~\ref{fig:hyperedges}.
However, by the modification of the Wang tile system $V$ above, we ensured that the graph $G_e$ contains exactly those edges shown in Figure~\ref{fig:hyperedges}.
Suppose one of the graphs $G_e$ would contain an edge $(x',y')$ which is not shown in Figure~\ref{fig:hyperedges}, then the tile corresponding to $y'$ could occur in two positions which are less than three rows and columns apart --- a property that was excluded by the modification.

We will show that $p(\cL(A)) = L$.
Firstly, consider an array $\bar p$ of tiles from $\Theta$ which is the Wang-tiled version of the picture $p\in\cL(W)$.
We will show that the SA-hypergraph automaton $A$ accepts a picture graph $M$ such that $p(M) = p$.
We assume $M$ to be embedded in $\Z^2$ such that the nodes cover the axis-parallel rectangle spanned by the points $(1,1)$ and $(h(p),w(p))$, every $v$-edge points downwards, and every $h$-edge points rightwards; recall that our coordinates represent the rows and columns of a matrix.
The derivation leading to the final configuration $(M,\es,g)$ simulates the assembly of tiles which form $\bar p$ as shown in Figure~\ref{fig:simulationmethod}.
The north and west edges of the tile $t_{TL} = \bar p_{(1,1)}$ in the top left corner of $\bar p$ are labelled by $\#$, and therefore, the node $x_{TL} = \theta^{-1}(t_{TL})$ corresponding to $t_{TL}$ forms an initial graph $M_0$.
The adjacent edges of two neighbouring tiles $s,t$ in $\bar p$ are labelled by the same colour.
Suppose $s$ is the west neighbour of $t$, then $\sigma_E(s) = \sigma_{W}(t) \neq \#$ and both tiles belong to the same row, implying that $\sigma_N(s)=\# \iff \sigma_N(t) = \#$ and $\sigma_S(s) = \# \iff  \sigma_S(t) = \#$.
Therefore, their corresponding nodes in $G$ are connected by an $h$-edge $(\theta^{-1}(s),\theta^{-1}(t))\in E_h$.
Analogously, if $s$ is the north neighbour of $t$, then $(\theta^{-1}(s),\theta^{-1}(t))\in E_v$.
Next, we see that the hyperedges of type~I and type~II can be used in order to create the top row and left column of the graph $M$, respectively.
Furthermore, the hyperedges of type~III can be used in order to create all the remaining nodes of $M$.
We conclude that $(M_0,\set{x_{TL}},id)\overset{*}{\underset{A}\rightarrow} (M,O,g)$ is a derivation in $A$ and we will prove that $(M,O,g)$ has to be a final configuration with $O=\es$.
Observe, that hyperedges of types~I and~II leave all the nodes active while hyperedges of type~III deactivate at least the top left node in the hyperedge.
Thus, all nodes except for those in the bottom row and in the right column will be deactivated in the configuration $(M,O,g)$.
Furthermore, in order to create the bottom row and right column hyperedges of type~III.2 and~III.3 are used, respectively, and one rule of type~III.1 is used in order to create the bottom-right node of $M$.
It is easy to see that the derivation function is designed such that all nodes will be deactivated in the configuration $(M,O,g)$ and, therefore, $A$ accepts $M$.

Now, let $M = (N_M,E_{v,M}\cup E_{h,M},\pi_M) \in L(A)$ be a graph which is generated by $A$.
Let $G$ be accepted by the derivation
\begin{equation*}
	(M_0,O_0,g_0) \underset A\to (M_1,O_1,g_1) \underset A\to \cdots \underset A\to (M_k,O_k,g_k)
\end{equation*}
where $(M_0,O_0,g_0) = (G_{e_0},O_{e_0},id)$ is an initial configuration with $e_0\in E_0$ and $(M_k,O_k,g_k) = (M,\es,g)$ is a final configuration.
Let $N_i$ be the node set of the graph $M_i$.
Note that for any $0\le i\le k$ the function $g_i$ is the restriction of $g$ by $N_i$, that is $g_i = g|_{N_i}$.
In order to avoid confusion, nodes in the graph $M$ are consistently denoted by $x,y$ and nodes in the graph $G$ are consistently denoted by $x',y'$; the nodes may have subscripts.

Let the nodes in the graphs $M_0,\ldots,M_k$ be embedded in $\Z^2$ such that all $h$-edges point rightwards and all $v$-edges point downwards; just like we did above.
The creation of graph $M = M_k$ starts with the initial graph $M_0$ which contains only one node $x_{TL}\in N_T\cap N_L$.
Let $x_{TL}$ lie on position $(1,1)$ in all of the graphs $M_0,\ldots,M_k$.
The graph $M_0$ can be extended rightwards by using hyperedges of type I and downwards by hyperedges of type II.
Since none of the hyperedges attach a new node upwards or leftwards of an existing node in $M_{i-1}$ in order to obtain $M_{i}$, the node $x_{TL}$ lies in the top row  and in the left column of $M_i$.
By the definition of type I and II hyperedges, for every node $y$ in the top row (resp., left column) of $M$ we have $g(y)\in N_T$ (resp., $g(y)\in N_L$).
By using hyperedges of type III the area spanned by the top row and left column can be filled with nodes.
It is easy to see that for all graphs $M_0,\ldots, M_k$ we have that if a node lies on position $(i,j)$, then for all $(i',j') \in [i]\times[j]$ a node lies on position $(i',j')$.
Furthermore, if $i'< i$, then the node on position $(i',j')$ has an outgoing $v$-edge, and if $j' <j$, then the node on position $(i',j')$ has an outgoing $h$-edge.
In other words, in the axis-parallel rectangle spanned by the points $(1,1)$ and $(i,j)$ all nodes are connected by edges with all direct neighbours (nodes which have an Euclidean distance of $1$).

In the final configuration $(M_k,\es,g)$ there is no active node.
Thus, the last node which is added to the graph $M_{k-1}$ in order to obtain $M_k$ is a node $x_{BR}$ such that $g(x_{BR})\in N_B\cap N_R$, as all other derivation rules will leave some nodes active.
Next, let us consider the nodes which belong to the same row and column as $x_{BR}$ does.
Note that if two nodes $x$ and $y$ in $M$ are connected by an edge, then the corresponding nodes $g(x)$ and $g(y)$ in $G$ are connected by an edge, too; more precisely, if $(x,y)\in E_{v,M}$, then $(g(x),g(y))\in E_v$, and if $(x,y)\in E_{h,M}$, then $(g(x),g(y))\in E_h$. 
Since a node in $N_B$ (resp., $N_R$) only is connected by $h$-edges (resp., $v$-edges) in $G$ to other nodes from $N_B$ (resp., $N_R$), we see that for every node $y$ in the row of $x_{BR}$ (resp., column of $x_{BR}$) we have $g(y)\in N_B$ (resp., $g(y)\in N_R$).
A node $y'\in N_B$ (resp., $y'\in N_R$) does not have any outgoing $v$-edges (resp., $h$-edges) as the south edge (resp., east edge) of $\theta(y')$ is labelled by $\#$.
We conclude that $x_{BR}$ sits in the bottom row and right column of the graph $M$ and, by the observations made above, this implies that $M$ is a picture graph.

We claim that the picture $p(M)$ which corresponds to the graph $M$ can be generated by the assembly $\bar p$ given by the embedding of nodes in $M$ and the function $\theta\circ g$.
Clearly, for every node $y$ on position $(i,j)$ in $M$ we have that $p(M)_{i,j} = \pi_M(y) = \lambda(\theta(g(y)))$, therefore, the pictures $p$ and $\lambda(\bar p)$ coincide.
Next, we prove that $\bar p$ is a tiled picture in the Wang tile system $W$.
Recall, that all nodes on the top, right, bottom, and left border of $M$ correspond to tiles in $M_T$, $M_R$, $M_B$, and $M_L$, respectively, and therefore, $\bar p$ is well-bordered.
Let $t_x$ and $t_y$ be two neighbouring tiles in $\bar p$ which lie on positions $(i,j)$ and $(i,j+1)$, respectively.
Let $x$ and $y$ be the nodes in $M$ which lie on the positions $(i,j)$ and $(i,j+1)$, respectively. 
Note that $t_x = \theta(g(x))$ and $t_y = \theta(g(y))$.
Since $M$ is a picture graph, $(x,y) \in E_{h,M}$ and $(g(x),g(y))\in E_h$.
The edge set $E_h$ was build to ensure that $\sigma_E(t_x) = \sigma_W(t_y)$.
We conclude that all adjacent east-west edges in $\bar p$ have matching colours. By symmetric arguments, we also conclude that all adjacent north-south edges in $\bar p$ have matching colours.
Therefore, $\bar p$ is a tiled picture in $W$ and $p \in L$.

Finally, let us consider the case when $N_T\cap N_B \neq \es$.
We can add a component to the SA-hypergraph automaton which works similar to a non-deterministic finite automaton and where every hyperedge induces an graph of type~I in Figure~\ref{fig:hyperedges}.
The initial graphs are given by all nodes from $N_T\cap N_B \cap N_L$.
For all nodes $x,y\in N_T\cap N_B$ with $(y,x)\in E_h$ we define a hyperedge $e_{x,y}$ such that $f(e_{x,y}) = \set{x,y}$.
The derivation function is given as $d(e_{x,y}) = \set y \to \set x$ if $x\notin N_R$, and $d(e_{x,y}) = \set y \to \es$ otherwise.
Obviously, this attachment to the hypergraph $A$ accepts all graphs which correspond to pictures $p\in L$ with $h(p) =1$.
The case when $N_L\cap N_R\neq \es$ can be covered analogously.
\qed
\end{proof}

\begin{figure}[ht]
\begin{center}
\includegraphics[width=0.60\textwidth]{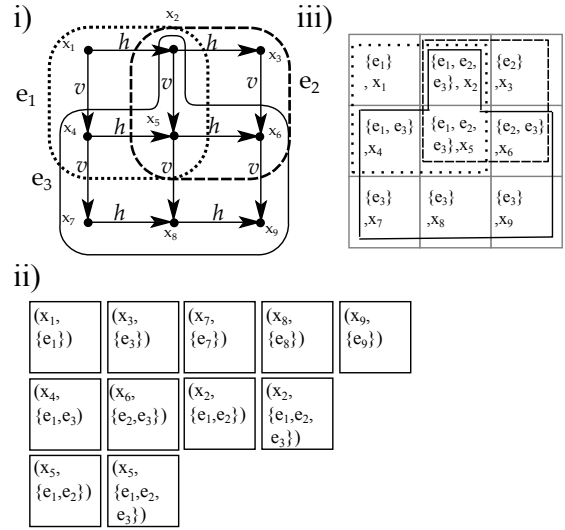}
\caption{
Let $A = (N,E,f,d,G,E_0)$ be a SA-hypergraph automaton where $N$, $E$, $f$, and $G$ are defined in part i). funtion $d$ is defined such that $d(e_1) = \{x_1\} \rightarrow \{x_2,x_4,x_5\}$, $d(e_2) = \{x_2,x_5\} \rightarrow \{x_2,x_5,x_6\}$ and $d(e_3) = \{x_2,x_4,x_5,x_6\} \rightarrow \{ \}$. SA-hypergraph automaton starts from $e_1$. Part ii) shows the set of all the possible tile candidates. On each tile related node and the set of $\psi$ are written. The tiling on part iii) is the result of overlapping of three hyperedges $e_1, e_2$ and $e_3$.}
\label{fig:HypergraphToWangTile}
\end{center}
\end{figure}

Next, we prove that a picture language $L = p(\cL(A))$, associated to the graph language $\cL(A)$, is WTS-recognizable if $A$ does not contain a strong loop.

Let $A$ be an SA-hypergrph automaton.
A series of hyperedges $s = \langle e_0, e_1, \ldots , e_n \rangle$ from $A$ is a {\em (derivation) loop} if $e_0 = e_n$ and $Q_{2,i} \cap Q_{1,i+1} \neq \emptyset$ where $d(e_i)=Q_{1,i} \rightarrow Q_{2,i}$ for $0 \le i < n$.
Loops in an SA-hypergraph automaton are a prerequisite for using a hyperedge several times in one derivation.
Therefore, an SA-hypergraph automaton without any loops can only accept a finite graph language.
Let $G_i = G_{e_i}$ be the graph induced by $e_i$, let $x$ be a node in $G_0=G_n$,
and let $O_i = Q_{2,i} \cap Q_{1,i+1}$ be set overlapping incoming/outgoing active nodes of $G_i$ and $G_{i+1}$.
There is a path in the underlying graph of $A$ from $x$ to $x$ which only visits the subgraphs $G_0,\ldots,G_n$, in the given order, and passes through at least one node of each $O_i$ (the path may use incoming and outgoing edges).
The loop $s$ is a {\it strong loop} if, on this path, the number of incoming horizontal edges equals the number of outgoing horizontal edges and the number of incoming vertical edges equals the number of outgoing vertical edges.
In other words, when starting from a configuration $M$ and successively gluing the hyperedges from $s$ to $M$, then the subgraph added by the hyperedge $e_0$ and the subgraph added by the hyperedge $e_n$ fully overlap when naturally embedded in $\Z^2$.
Note that, by Remark~\ref{rem:subgrid}, all graphs $G_i$ are subgrids which implies that the choice of the path from $x$ to $x$ does not matter in this definition.

\begin{theorem}\label{thm:two}
Let $A$ be a SA-hypergraph automaton without any strong loops.
The picture language $L = p(\cL(A))$, associated to the graph language $\cL(A)$, is WTS-recognizable (Wang tile system recognizable).
\end{theorem}

\begin{proof}
Let $A=(N,E,f,d,G, E_0)$ and let $G=(N,E_v\cup E_h,\pi)$.
We may assume that $e\in E_0$ if and only if $d(e) = \es \to O_e$. Therefore, none of the initial hyperedges can be used in a transition. This assumption is justified by the fact that we can duplicate all hyperedges in $E_0$ such that one copy can be used in a transition  but does not belong to $E_0$ and the other copy which belongs to $E_0$ cannot be used in a transition.
Furthermore, any hyperedge without incoming active nodes which does not belong to $E_0$ is useless and can be removed from $E$.

For a node $x \in N$ we define the list of {\em related hyperedges} to $x$, $H_x= \sett{ e \in E}{ x \in f(e) }$.
Let $x$ be a node and $\psi\sse H_x$.
We call a hyperedge $g\in \psi$ a {\em generator} of $(x,\psi)$ if $x\notin Q_1$ with $d(g) = Q_1\to Q_2$.
Note that if $g\in E_0$, then $g$ must be a generator.
We call a hyperedge $c\in \psi$ a {\em consumer} of $(x,\psi)$ if $x\notin Q_2$ with $d(c) = Q_1\to Q_2$.
The pair $(x,\psi)$ is a {\em tile candidate} if $\psi$ contains exactly one generator $g_{(x,\psi)}$ and exactly one consumer $c_{(x,\psi)}$; furthermore, if $g_{(x,\psi)} = c_{(x,\psi)}$, we require that $\psi = \set{g_{(x,\psi)}}$.
Note that if  $g_{(x,\psi)} \neq c_{(x,\psi)}$, then for all $e\in \psi$ with $d(e) = Q_1\to Q_2$, we have that $x\in Q_1$ unless $e$ is the generator and $x\in Q_2$ unless $e$ is the consumer.
The tile candidate $(x,\psi)$ describes the attachment of a copy of the node $x$ to the output graph by the generator; afterwards, $x$ is used as active node by all hyperedges in $\psi\sm\set{g_{(x,\psi)},c_{(x,\psi)}}$; finally, $x$ is deactivated by the consumer.
Let $G_\psi$ be the node-induced subgraph of $G$ by $\bigcup_{e\in\psi}f(e)$.
If $G_\psi$ is not a subgrid (a subgraph of some picture graph), we remove $(x,\psi)$ from the set of tile candidates.
Let $\Psi$ denote the set of all remaining tile candidates.

The Wang tile system $W = (\Sigma, C, \Theta)$ which recognizes $L$ is constructed based on the list $\Psi$.
In order to recognize the picture language associated to $\cL(A)$, we have to define the attachments of tile candidates. We use unordered pairs $\set{(x,\psi),(y,\phi)}\in\Psi^2$ of tile candidates for the colours on the edges.
For a tile candidate $(x,\psi)\in \Psi$ we define the set of labelled Wang tiles
$$\Theta_{(x,\psi)} = \mathcal{S}_{N,{(x,\psi)}} \times \mathcal{S}_{E,{(x,\psi)}}\times \mathcal{S}_{S,{(x,\psi)}}\times\mathcal{S}_{W,{(x,\psi)}}\times \set{l_x}$$
where $l_x$ is the label $\pi(x)$ and $\mathcal{S}_{N,{(x,\psi)}}$, $\mathcal{S}_{E,{(x,\psi)}}$, $\mathcal{S}_{S,{(x,\psi)}}$, $\mathcal{S}_{W,{(x,\psi)}}$ are sets of colours which are defined below.
The set of all tiles is the union $\Theta = \bigcup_{(x,\psi)\in\Psi} \Theta_{(x,\psi)}$.

For $(x,\psi),(y,\phi)\in \Psi$, we let $\set{(x,\psi), (y,\phi)} \in \mathcal{S}_{E,(x,\psi)}$ and $\set{(x,\psi), (y,\phi)} \in \mathcal{S}_{W,(y,\phi)}$ if and only if
\begin{enumerate}%
\item $( x, y ) \in E_h$,
\item $H_x \cap \phi \subseteq \psi$,
\item $\psi \cap H_{y} \subseteq \phi$, and
\item $g_{(x,\psi)} = g_{(y,\phi)}$ or $y\in Q_1$ for $d(g_{(x,\psi)}) = Q_1 \to Q_2$ or
	$x\in Q_1'$ for $d(g_{(y,\phi)}) = Q_1' \to Q_2'$.
\end{enumerate}
For $(x,\psi),(y,\phi)\in \Psi$, we let $\set{(x,\psi), (y,\phi)} \in \mathcal{S}_{S,(x,\psi)}$ and $\set{(x,\psi), (y,\phi)} \in \mathcal{S}_{N,(y,\phi)}$ if and only if $( x, y) \in E_v$ and conditions~2 to~4 are satisfied.
For $(x,\psi)\in \Psi$, we let $\mathcal{S}_{E,{(x,\psi)}}=\set\#$ if $x$ does not have an incoming vertical edges in the graph $G_\psi$. By symmetric condition we let  $\mathcal{S}_{N,{(x,\psi)}}=\set\#$, $\mathcal{S}_{S,{(x,\psi)}}=\set\#$, or $\mathcal{S}_{W,{(x,\psi)}}=\set\#$.

Now, consider an $m\times n$-picture graph $M = (N_M,E_{v,M}\cup E_{h,M},\pi_M)\in \cL(A)$.
We will show that there is a tiled version $\bar p$ of picture $p = p(M)$ which uses tiles from $\Theta$ and, therefore, $p$ is recognized by $W$.
Let $G$ be accepted by the derivation
\[
        (M_0,O_0,g_0) \underset A\to (M_1,O_1,g_1) \underset A\to \cdots \underset A\to (M_k,O_k,g_k)
\]
where $(M_0,O_0,g_0) = (G_{e_0},O_{e_0},id)$  is an initial configuration (that is $e_0\in E_0$) and $(M_k,O_k,g_k) = (M,\es,g)$ is a final configuration.
Let $e_i$ be the hyperedge and $P_i\sse O_{i-1}$ be the active nodes which are used in the transition $(M_{i-1},O_{i-1},g_{i-1}) \underset A\to (M_{i},O_{i},g_{i})$. Let $d(e_i) = Q_{1,i}\to Q_{2,i}$ for $1\le i\le k$.
Recall that, by definition, $M_{i-1}$ is a full subgraph of $M_{i}$ and, by induction, every graph $M_i$ is a full subgraph of $M$.
Being an $m\times n$-picture graph, the nodes in $M$ can be naturally embedded in $[m]\times [n]$ by the function $f_M$.

Consider one node $x'\in N_M$ and its original $x = g(x')$ in $G$.
We assign to $x'$ a list of hyperedges $\psi\sse E$ such that $e_i\in\psi$ if $x'\in P_i$ or $x'$ belongs to $M_{i}$ but not $M_{i-1}$.
We intend to use a tile from $\Theta_{(x,\psi)}$ for the pixel ${\bar p}_{f_M(x')}$ representing $x'$ in the tiled picture $\bar p$.
Observe that $\psi$ contains a consumer as $x'$ is not active in the final configuration and $\psi$ cannot contain two consumers because a node can only be deactivated once.
In addition, the hyperedge $e_i$ such that $x'$ belongs to $M_{i}$ but not $M_{i-1}$ is the single generator in $\psi$.
Since $G_\psi$ is isomorphic to a subgraph of $M$, we conclude that $(x,\psi)$ is indeed a tile candidate.
If $x'$ does not have an outgoing horizontal edge, then the node $x$ in the graph $G_\psi$ cannot have an outgoing horizontal edge either and, therefore, $\mathcal{S}_{E,(x,\psi)} = \set\#$. Symmetric arguments apply if $x$ does not have an incoming horizontal, outgoing vertical, or incoming vertical edge.

Next, consider two nodes $x',y'\in N_M$ which are connected by an edge and, by symmetry, assume $(x',y')$ is a horizontal edge.
Let $x = g(x')$, $y = g(y')$ be their originals and let $\psi,\phi$ be the set of hyperedges associated to $x',y'$, respectively.
We will show that $\{(x,\psi),(y,\phi)\}$ is a colour in $\mathcal{S}_{E,(x\psi)}$ as well as in $\mathcal{S}_{W,(y,\phi)}$.
Thus, we can choose tiles from $\Theta_{(x,\psi)}$ and $\Theta_{(y,\phi)}$ for the positions $f_M(x')$ and $f_M(y')$ in $\bar p$, respectively.
Clearly, the choice of the tiles also depends on the other neighbours of $x'$ and $y'$.
We have to show that conditions~1 to~5, above, are satisfied.
The first condition is satisfied by assumption.
By contradiction, suppose the second condition is not satisfied.
There is $e_i\in H_x \cap \phi \setminus \psi$; thus, in the $i$-th step of the derivation we use the hyperedge $e_i$ that presupposes or generates an edge $(x'',y')$ in $M$ where $g(x'') = x$ but $x'' \neq x'$.
This would imply that $y$ has two incoming horizontal edges whence $M$ is not a picture graph.
The third condition is satisfied by symmetric arguments.
The edge $(x',y')$ in $M$ can only be created in a step $i$ where $x'$ or $y'$ is added to the graph $M_{i-1}$.
Thus, $x'$ and $y'$ either have the same generator in $(x,\psi)$ and $(y,\phi)$, or $x'$ is in the active nodes when $y'$ is generated, or $y'$ is in the active nodes when $x'$ is generated.
In all cases condition~4 is satisfied.

We conclude that a tiled picture $\bar p$ such that $p = p(M)$ and $M\in \mathcal{L}(A)$ can be generated by using tiles from $\Theta$ and, therefore, $p(M)\in\cL(W)$.

\bigskip

Consider a picture $p\in \cL(W)$ and let $\bar{p}$ be the tiled version of $p$, using tiles from $\Theta = \bigcup_{(x,\psi)\in\Psi} \Theta_{(x,\psi)}$.

We start by introducing the concept of masks which can be seen as connected subpictures of the tiled picture $\bar p$ that represent the nodes in one hyperedge.
A \textit{mask} $\mathfrak{m}$ is a $h(\bar{p}) \times w(\bar{p})$ matrix of tiles from $\Theta \cup \{ \text{\it{empty}}\}$, such that either $\mathfrak{m}_{(i,j)} = \text{\it{empty}}$ or $\mathfrak{m}_{(i,j)} = \bar{p}_{(i,j)}$ for all $(i,j)\in [h(\bar{p})] \times [w(\bar{p})]$.
In addition, we require that the non-empty entries in $\mathfrak m$ are connected; that is, for  every pair of tiles $\mathfrak{m}_{(i',j')},\mathfrak{m}_{(i,j)}\in \Theta$ there exists a sequence $r=\langle r_0, r_1, \cdots , r_n\rangle$ of tiles in $\mathfrak{m}$ such that $r_0 = \mathfrak{m}_{(i,j)}$, $r_n = \mathfrak{m}_{(i',j')}$, $r_k \in \Theta$, and $r_k$, $r_{k+1}$ must be adjacent for all $0 \leq k < n$.

Let $e\in E$ be an hyperedge and let $G_e = (N_e,E_{e,v}\cup E_{e,h}, \pi_e)$ be the graph induced by this hyperedge.
By Remark~\ref{rem:subgrid}, we assume that $G_e$ is a subgrid.
We say $G_e$ is mapped to a mask $\mathfrak m$ if there is a injective function $h\colon N_e \to [h(\bar p)]\times [w(\bar p)]$ which satisfies:
$\mathfrak m_{(i,j)}$ belongs to $\Theta$ if and only if $(i,j)$ is in the domain of $h$;
for all nodes $x,y\in N_e$ there is an edge $(x,y)\in E_{e,h}$ (resp., $(x,y)\in E_{e,v}$) if and only if $h(x)$ is in north (resp., west) neighbour of $h(y)$.
Whenever we use this mapping, we will ensure that for all $x\in G_e$ the tile $\bar p_{h(x)}$ belongs to $\Theta_{(x,\psi)}$ for some $\psi\sse E$.
  
Consider a tile $t\in\bar p_{(i,j)} \in \Theta_{(x,\psi)}$ and a hyperedge $e\in \psi$.
We define the mask $\mathfrak m^{[(i,j),x,e]}$ such that the graph $G_e$ can be mapped by function $h$ to $\mathfrak m^{[(i,j),x,e]}$ and $h(x) = (i,j)$.
We say that $e$ is the hyperedge related to the mask $\mathfrak m^{[(i,j),x,e]}$.
Let $t' =\bar p_{(i',j')} \in \Theta(y,\phi)$ be a tile that is adjacent to $t$ and let $e \in \psi$.
For simplicity we only consider the case when $t'$ is the east neighbour of~$t$; i.e., $(i',j') = (i+1,j)$.
We will show that if $(i',j')$ is non-empty in $\mathfrak m^{[(i,j),x,e]}$, then $e \in \phi$.
Since $t'$ is the east neighbour of $t$ conditions~1 to~4, above, apply.
As $(i',j')$ is non-empty in $\mathfrak m^{[(i,j),x,e]}$, there exists a horizontal edge $(x,z)$ in $G_e$.
Furthermore, from conditions~1 and~4 it follows that $(x,y)$ is a horizontal edge in the graph $G_{g}$ induced by the generator $g = g_{(x,\phi)}$.
As both graphs $G_e$ and $G_g$ are subgraphs of the subgrid $G_\psi$, we see that the edges $(x,y)$ and $(x,z)$ coincide, thus, $y = z$.
We conclude $y\in G_e$ and $e\in H_y$.
By condition~3, $e\in \phi$.
Because the hyperedge $e$ induces a connected graph, we can infer that for all non-empty $\mathfrak m^{[(i,j),x,e]}_{(i'',j'')} \in \Theta_{(z,\chi)}$, we find $e\in \chi$.
Note that this also implies that $\mathfrak m^{[(i,j),x,e]} = \mathfrak m^{[(i',j'),y,e]}= \mathfrak m^{[(i'',j''),z,e]}$.

We define the set of masks $\mu = \{ \mathfrak m^{[(i,j),x,e]} | \bar p_{(i,j)}\in\Theta_{(x,\psi)}, e\in\psi \}$ which are induced by hyperedges in the above manner.
Intuitively, every mask in $\mu$ represents one transition in the derivation of a picture graph $M$ which represents the picture $p = p(M)$.
In order to use a transition defined by a mask, we need to guarantee that all of its input areas exist and are active. 
We will order the set $\mu$ accordingly.
Let us define the relation ${R} \subseteq \mu \times \mu$ such that $(\mathfrak m,\mathfrak n)\in R$ if the transition represented by $\mathfrak m$ has to be used before the transition represented by $\mathfrak n$. 
Let $\mathfrak m$ and $\mathfrak n$ be two distinct masks in $\mu$.
The pair $(\mathfrak m,\mathfrak n)$ is in ${R}$ if there exists 
$(i,j)$ such that $\mathfrak m_{(i,j)} = \mathfrak n_{(i,j)} \in \Theta_{(x,\psi)}$, and
$\mathfrak m = \mathfrak m^{[(i,j),x,g]}$ where $g = g_{(x,\psi)}$
or $\mathfrak n = \mathfrak m^{[(i,j),x,c]}$ where $c = c_{(x,\psi)}$.
The pair ${(\mu,{R})}$ can be seen as directed graph $G_{\mu}$.
First, we show that the graph ${G_{\mu}}$ does not contain any loops, afterwards, a topological sort of this graph is used to order the transitions represented by the masks. 


When two or masks overlap on a tile (have a common non-empty entry), regarding the construction of tile candidates, we know that the related hyperedge of exactly one of these masks is the generator of the input area of the other hyperedges. Hence, these masks are connected in the graph  ${G_\mu}$. 
By contradiction, assume that $\langle \mathfrak n_0,\mathfrak n_1,\ldots,\mathfrak n_l \rangle$ is a non-trivial loop in ${G_{\mu}}$ (i.e., $(\mathfrak n_i,\mathfrak n_{i+1})\in R$ for every $0\le i<l-1$ and $(\mathfrak n_l,\mathfrak n_{0})\in R$).
However, the sequence of related hyperedges to this sequence of mask is a strong loop in the SA-hypergraph automaton $A$ which was excluded by assumption.
Moreover, since two tiles with different generators cannot connect without satisfying conditions~4, the graph $G_\mu$ must be connected.
Therefore, graph  ${G_{\mu}}$ can be topologically sorted. Sorting of the hyperedges guaranteed that the active input nodes of one hyperedge are generated before the gluing of the hyperedge. 

By contraction, assume that graph ${G_\mu}$ has two distinct nodes $\mathfrak m_1$ and $\mathfrak m_2$ without any input edges. 
Let $\mathfrak m_3$ be the first node in the topological order such that  paths $\mathfrak m_1 \to^* \mathfrak m_3$ and $\mathfrak m_2 \to^* \mathfrak m_3$ exist in $\mathcal G_\mu$.
As $\mathfrak m_3$ is chosen minimal, these paths do not share any node other than $\mathfrak m_3$.
Recall that all incoming active nodes of a hyperedge are connected. Considering that two nodes cannot connect to each other unless they are in the same hyperedge or they have glued to each other, we have a contradiction as $\mathfrak m_3$ cannot be the first common node on both paths. 
We conclude that graph ${G_{\mu}}$ has only one node without input.

Now, let $\mathfrak m_0, \mathfrak m_1, \ldots, \mathfrak m_k$ be the topological sort of $\mu$ by the relation $\mathcal R$.
We define $\mathfrak m + \mathfrak n = \mathfrak o$ such that $\mathfrak o_{(i,j)}=empty$ if $\mathfrak m_{(i,j)} = \mathfrak n_{(i,j)} = empty$; otherwise, $\mathfrak o_{(i,j)} = \bar p_{(i,j)}$.
We will show that a graph $M_k$ can be generated by a derivation
\[
  (M_0,O_0,g_0) \underset A\to (M_1,O_1,g_1) \underset A\to \cdots \underset A\to (M_k,O_k,g_k)
\]
such that the graph $M_i$ can be mapped to the mask $\sum_{j=0}^i \mathfrak m_j$;
this implies that $m_k$ can be mapped to $ \bar p = \sum_{j=0}^k \mathfrak m_j$.
Let $e_i$ be the hyperedge related to the mask $\mathfrak m$.
The graph $M_0 = G_{e_0}$ is an initial graph because $\mathfrak m_0$ has no incoming edges in $\mathcal G_\mu$ and, therefore, the derivation function of $e_0$ is $d(e_0) = \emptyset \to Q_2$;
thus, $(M_0,O_0,g_0)$ where $O_0 = Q_2$ and $g_0 = id$ is an initial configuration.
In derivation step $(M_{i-1},O_{i-1},g_{i-1}) \underset A\to (M_i,O_i,g_i)$ we use the hyperedge $e_i$.
By induction, we can assume that $M_{i-1}$ can be mapped to $\sum_{j=0}^{i-1} \mathfrak m_j$ by a function $h_{i-1}$.
There is only one way to glue the hyperedge $e_i$ to $M_{i-1}$ such that resulting graph $M_i$ can be mapped to $\sum_{j=0}^{i} \mathfrak m_j$.
We have to prove that all incoming active nodes of $G_{e_i}$ exist and are active in $M_i$.
Let $x$ be an incoming active node which is represented by the tile $\bar p_{(a,b)}\in \Theta_{(x,\psi)}$. The definition of $R$ ensures that the mask representing the generator of $(x,\psi)$ in $(a,b)$ has already been used and that the mask representing the consumer of $(x,\psi)$ in $(a,b)$ has not yet been used.
Finally, every tile candidate has a consumer which means that there are no active nodes in the final configuration $(M_k,O_k,g_k)$.
As result, the picture $p$, generated by the suggested tiling system, is in $p(\cL(A))$.  \qed
\end{proof}

\section{Conclusion}
We introduced SA hypergraph automata, a  language/automata theoretic model for patterned self-assembly systems.
SA hypergraph automata accept all recognizable picture languages but, unlike other models,
 (e.g. Wang Tile Automata) SA-hypergraph automata do not  rely on an {\it a priori} given  scanning strategy of a
 picture. This property makes the  SA hypergraph automata  better suited  to model 
 DNA-tile-based self-assembly systems. 

SA-hypergraph automata   provide a  natural automata-theoretic model for patterned
 self-assemblies  that will  enable us to analyse self-assembly  in an automata-theoretic framework. 
This framework  lends itself easily to, e.g.,  descriptional and computational complexity analysis, 
and  such studies  may ultimately lead to  classifications and hierarchies of patterned self-assembly
 systems based on the properties of their corresponding SA-hypergraph automata. An additional feature is
 that each SA-hypergraph automaton  accepts an entire class of ``supertiles'' as opposed to a singleton set,
 which may also be of interest  for some applications or analyses.

\bibliographystyle{abbrv}
\bibliography{reflist}

\end{document}